\documentclass[a4paper,11pt]{article}
\usepackage{jheppub,xcolor,slashed}

\newtheorem{theorem}{Theorem}[section]
\newtheorem{corollary}{Corollary}[theorem]
\newtheorem{lemma}[theorem]{Lemma}

\DeclareMathOperator{\Ds}{\slashed{\cal{D}}}

\def \be {\begin{equation}}
\def \ee {\end{equation}}
\def \bea {\begin{eqnarray}}
\def \eea {\end{eqnarray}}

\def \rR {{\rm R}}
\def \rNS {{\rm NS}}

\newcommand\HR[1]{{\it \color{purple}\bf  [#1 - HR]}}

\title{\boldmath The BTZ black hole, the Third Law and gravitational collapse
}

\author{Aidan M. McSharry, Harvey S. Reall and Jorge E. Santos}
\affiliation{Department of Applied Mathematics and Theoretical Physics, University of Cambridge\\
Wilberforce Road, Cambridge CB3 0WA, United Kingdom}

\emailAdd{amm323@cam.ac.uk, hsr1000@cam.ac.uk, jss55@cam.ac.uk}

\abstract{
We prove that a Third Law of black hole mechanics holds for the BTZ black hole in 3d gravity coupled to matter obeying the dominant energy condition. This law states that initial data containing a trapped surface cannot evolve in finite time to a black hole that coincides with extremal BTZ on the event horizon, for any choice of boundary conditions at infinity. The result is proved by defining a quasilocal energy and angular momentum and using spinorial methods to establish a BPS inequality that is saturated by the extremal BTZ solution. Nevertheless, we show that gravitational collapse of a massless scalar field can result in the formation of an extremal BTZ black hole in finite time. The difference between these results arises because gravitational collapse occurs in the Neveu-Schwarz sector, whereas extremal BTZ is supersymmetric in the Ramond sector. Solutions that settle down to extremal BTZ in infinite time along the event horizon are also discussed.
}

\begin{document}
\maketitle
\flushbottom

\section{Introduction}
\label{sec:intro}

The Third Law of black hole mechanics asserts that, in classical General Relativity with physically reasonable matter, it is impossible for a non-extremal black hole to become extremal in finite time. This was conjectured in \cite{Bardeen:1973gs} and a proof was presented in \cite{Israel:1986gqz}. However, an error in this proof was identified in \cite{Kehle:2022uvc}, where counterexamples to the Third Law were constructed. These counterexamples are solutions of Einstein-Maxwell theory coupled to a massless charged scalar field. They describe spherically symmetric gravitational collapse to form a black hole that contains a region that is exactly Schwarzschild, inside the Schwarzschild apparent horizon, but then evolves to an exactly extremal Reissner-Nordstr\"om black hole on and outside the event horizon, in finite advanced time. Similar solutions exist for Einstein-Maxwell theory coupled to charged Vlasov matter with a large charge to mass ratio \cite{Kehle:2024vyt}. 

Subsequent work has proved that such solutions do {\it not} exist if Einstein-Maxwell theory is coupled to matter satisfying a certain upper bound on its charge density in terms of its energy density \cite{Reall:2024njy,McSharry:2025iuz}. The proof of this is spinorial, and relies on the fact that the extremal Reissner-Nordstr\"om solution admits a ``supercovariantly constant'' spinor. A similar result holds for supersymmetric Kerr-Newman-AdS$_4$ black holes \cite{McSharry:2025iuz}. This raises the question of whether the validity of the Third Law for a given theory might depend on the type of matter present in that theory. However, it has been conjectured \cite{Kehle:2022uvc} that there exist third-law violating solutions of {\it vacuum} gravity, describing a Schwarzschild black hole that evolves in finite time to an extremal Kerr black hole through absorption of gravitational waves. Recently it has been shown that the analogous conjecture is true in {\it five} spacetime dimensions \cite{Crump:2026kgu}. These results strongly suggest that it is unlikely that there exist 4d (or 5d) theories for which a Third Law holds for all extremal black holes, but there do exist classes of theories for which a Third Law holds for a subset of extremal black holes.

So far we have focused on the Third Law, the question of whether a non-extremal black hole can evolve to become an extremal black hole in finite time. Another question is whether an extremal black hole can form in finite time in gravitational collapse. Although similar, these two questions are independent. For example, it might be impossible to form a certain type of extremal black hole in gravitational collapse but possible to form it from a pre-existing two-sided non-extremal black hole. Or maybe it is impossible for a non-extremal black hole to evolve to a certain type of extremal black hole but nevertheless such an extremal black hole can form in collapse. However, in the situations considered in \cite{Kehle:2022uvc,Kehle:2024vyt,Reall:2024njy,McSharry:2025iuz,Crump:2026kgu} the answers to these two questions coincide. For \cite{Kehle:2022uvc,Kehle:2024vyt,Crump:2026kgu} the answer is positive and for \cite{Reall:2024njy,McSharry:2025iuz} it is negative. 

In this paper we will show that the answers to these questions are {\it different} for gravity in 3d with negative cosmological constant, which admits the BTZ family of black hole solutions \cite{Banados:1992wn}. We will adapt the methods of \cite{Reall:2024njy,McSharry:2025iuz} to prove a Third Law for this theory:

\begin{theorem} \emph{(Third Law for BTZ.)}
\label{thm:3rdlaw}
Consider a 3d smooth spacetime satisfying Einstein's equation with negative cosmological constant and matter satisfying the dominant energy condition. Consider a smooth spacelike hypersurface $\Sigma$ with the topology of an annulus with outer boundary $S$ and inner boundary $T$ (see Fig. \ref{fig:btz third law}). Assume that the future-directed ingoing null geodesics normal to $S$ are strictly converging and that $T$ is a weakly future outer trapped circle.\footnote{Weakly future outer trapped means that the future-directed outgoing null geodesics normal to $T$ have non-positive expansion, where ``outgoing'' on $T$ (and ``ingoing'' on $S$) means ``pointing into $\Sigma$''.}
Then the spacetime metric and extrinsic curvature on $S$ cannot coincide with those on a cross-section of the event horizon of an extremal BTZ black hole.
\end{theorem}
Following \cite{Israel:1986gqz} we assume that the notion of ``non-extremal'' implies the presence of a trapped surface.\footnote{This does not preclude the possibility that an extremal black hole might also have a trapped surface e.g. deep in its interior, as in \cite{Kehle:2022uvc}.} So this theorem asserts that if a black hole is non-extremal at some time then it cannot evolve to become extremal BTZ at the event horizon in finite time. See Fig \ref{fig:btz third law}. The theorem is independent of any choice of boundary conditions at timelike infinity.

\begin{figure}
    \centering
    \includegraphics[width=0.6\linewidth]{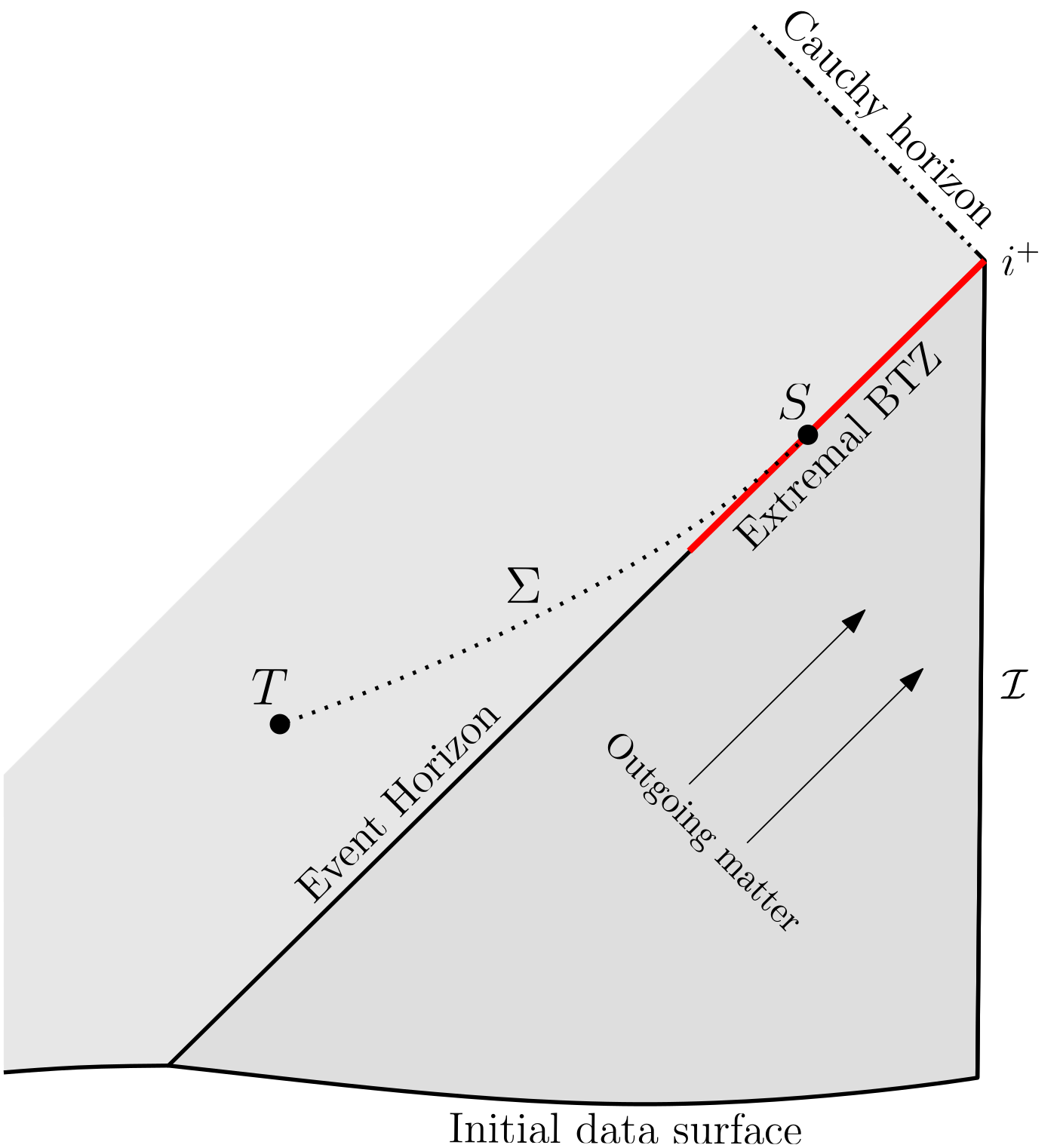}
    \caption{Penrose diagram showing a Third Law violating spacetime in which an initially non-extremal black hole evolves in finite advanced time to form a black hole that is exactly extremal BTZ at the event horizon.
    Theorem \ref{thm:3rdlaw} proves that this spacetime cannot exist if matter obeys the dominant energy condition. 
    The red line represents the portion of the event horizon on which the metric and extrinsic curvature of any cross-section are diffeomorphic to those of an extremal BTZ horizon cross-section. Here $S$ is a cross-section of this portion of the event-horizon, $T$ is a spacelike, outer trapped surface and $\Sigma$ is a spacelike hypersurface.
    The arrows denote outgoing matter. (We do not assume anything about boundary conditions at infinity so this need not be reflected back into the hole.) This figure is schematic as we do not assume axisymmetry.}
    \label{fig:btz third law}
\end{figure}

The proof of Theorem \ref{thm:3rdlaw} introduces ideas that may have other applications. Similarly to \cite{Reall:2024njy}, we will present spinorial definitions of {\it quasilocal} energy and angular momentum for 3d gravity with a negative cosmological constant. Given a spacelike circle $S$, these quantify the total energy and angular momentum ``contained within'' $S$. For a circle of axisymmetry in a BTZ black hole spacetime, these recover the usual BTZ mass and angular momentum parameters. For a general axisymmetric spacetime they also agree with standard definitions \cite{Brown:1986nw} of energy and angular momentum of asymptotically AdS$_3$ spacetimes when $S$ is taken to infinity. Importantly, we will show that, under certain conditions, our quasilocal quantities satisfy a BPS bound. More precisely, we define two sets of quasilocal quantities, referred to as ``Ramond'' and ``Neveu-Schwarz'', depending on whether certain spinors are periodic or anti-periodic around $S$. We prove that each set of quantities satisfies a BPS bound under suitable circumstances. 

The proof of Theorem \ref{thm:3rdlaw} is based on the fact that a circle in the extremal BTZ spacetime saturates our Ramond sector quasilocal BPS bound. If one tries to use similar methods to exclude spacetimes describing gravitational collapse to form an extremal BTZ black hole in finite time then one runs into the following problem. In this situation, one would take $\Sigma$ to be a disc with boundary $S$, and no inner boundary. In this case, to prove a BPS bound we need the spin structure on $S$ to extend to the disc, which is only true for the anti-periodic spin structure. Hence in this setup we can prove a BPS bound only for the Neveu-Schwarz quasilocal quantities. But this bound is {\it not} saturated by extremal BTZ. In other words, gravitational collapse takes place in the Neveu-Schwarz sector but extremal BTZ is supersymmetric only in the Ramond sector. Hence spinorial arguments do not exclude the existence of solutions describing gravitational collapse to extremal BTZ in finite time. 

We will show that solutions {\it do} exist. We will adapt the methods of \cite{Kehle:2022uvc} to construct solutions that describe the gravitational collapse of a massless scalar field to form an exactly extremal BTZ black hole in finite advanced time. Theorem \ref{thm:3rdlaw} implies that these solutions do not possess trapped circles. This makes these solutions candidates for ``critical'' solutions: solutions on the boundary in moduli space between solutions that form black holes and solutions that do not.\footnote{This is simplest to study for the case of ``transparent'' boundary conditions at infinity, for which non-black hole solutions are expected to disperse. With reflecting boundary conditions one would have to think more carefully about how to define criticality.} 
For some models it has been shown that part of this boundary corresponds to the extremal black hole threshold \cite{Murata:2013daa,Kehle:2024vyt,Angelopoulos:2026bez}. This part of the boundary is a hypersurface in moduli space: the critical surface. Solutions on this hypersurface describe gravitational collapse to form a black hole which ``settles down'' to an extremal black hole. In models with particle-like matter this settling down can occur in finite time \cite{Kehle:2024vyt} but with wave-like matter it generically occurs in infinite time \cite{Murata:2013daa, Angelopoulos:2026bez}. General arguments imply that a critical solution cannot possess a trapped surface. 

Our solutions have the property of forming an exactly extremal black hole in finite time. However, one would expect a ``generic'' critical solution to settle down to an extremal black hole only asymptotically, at late time, as first seen in \cite{Murata:2013daa}. 
We will present evidence for the existence of such solutions in 3d, assuming that the late time rate of decay of the scalar field along the event horizon coincides with that of a test field in the extremal BTZ spacetime. With the same assumption, we can also exclude the existence of solutions that ``asymptotically violate'' the Third Law, i.e., solutions with an initial trapped circle that asymptotically settle down to extremal BTZ.

This paper is organized as follows. Section \ref{sec:thirdlaw} introduces spinorial definitions of quasilocal energy and angular momentum and uses them to prove Theorem \ref{thm:3rdlaw}. Section \ref{sec:gluing} describes our construction of solutions describing gravitational collapse of a massless scalar to form an extremal BTZ black hole and the discussion of solutions approaching extremal BTZ asymptotically at late time. Section \ref{sec:discussion} contains further discussion of our results.

\section{Third Law}

\label{sec:thirdlaw}

\subsection{Spinor preliminaries}

We assume that spacetime is smooth, three-dimensional, oriented and time-oriented\footnote{this is sufficient to ensure that any spacetime we consider admits a spin structure, as every orientable, time-orientable 3d spacetime is spinnable (i.e. admits at least one spin structure) \cite{Geroch:1968zm, parker_theorems_1984}.}. The metric has positive signature. Latin indices $a,b,c$ are abstract indices, Greek indices refer to an orthonormal basis. Latin indices $i,j,k$ refer to spacelike basis vectors. We use units so that the Newton constant is $G=1$ and assume a negative cosmological constant
\be
 \Lambda = - \frac{1}{\ell^2}
\ee
where $\ell>0$ is the radius of curvature of the AdS$_3$ solution. We will work with $2$-component complex spinors. The gamma matrices satisfy 
\be
 \{ \Gamma^\alpha, \Gamma^\beta \} = -2 \eta^{\alpha \beta}
\ee
where $\eta^{\alpha\beta} = {\rm diag}(-1,1,1)$, $\Gamma^0$ is Hermitian and $\Gamma^i$ ($i=1,2$) anti-Hermitian. As usual we define
\be
 \bar{\psi} = \psi^\dagger \Gamma^0
\ee
In 3d there are two inequivalent, irreducible representations $\pm \Gamma^\mu$ for the gamma matrices. We choose
\be
 \Gamma^0 = \sigma^2 \qquad \Gamma^1 = i\sigma^3 \qquad \Gamma^2 = i\sigma^1
\ee
where $\sigma^i$ are the Pauli matrices.
This representation satisfies $\Gamma^2\Gamma^1\Gamma^0 = i$.
Associated with the existence of the inequivalent representations for the gamma matrices are two ``supercovariant'' derivatives \cite{Coussaert:1993jp}
\be
\label{supercovdef}
 \nabla^\pm_\mu \equiv \nabla_\mu \pm \frac{i}{2\ell} \Gamma_\mu
\ee 
where $\nabla_\mu$ is the Levi-Civita connection.

We also use the notation
\begin{equation}
\Gamma_{\mu\nu}\equiv \Gamma_{[\mu}\Gamma_{\nu]}
=\frac{1}{2}[\Gamma_\mu,\Gamma_\nu]\,,
\end{equation}
where antisymmetrisation is with unit weight.

\subsection{Quasilocal quantities}\label{sec:quasilocal quantities}

Let $\Sigma$ be a smooth compact spacelike hypersurface with boundary $\partial \Sigma$. We choose an orthonormal basis $\{e_0^a,e_1^a,e_2^a \}$ such that $e_0^a$ is the future-directed unit normal to $\Sigma$ and, on $\partial \Sigma$, $e_1^a$ is tangent to $\partial \Sigma$, and $e_2^a$ is the outward unit normal to $\partial \Sigma$ within $\Sigma$. On any connected component $C$ of $\partial \Sigma$ we define the hermitan matrix
\be
 \epsilon = \Gamma^2 \Gamma^0
\ee
with eigenvalues $\pm 1$. A spacetime spinor on $C$ can be decomposed into two spinors with opposite eigenvalues w.r.t. $\epsilon$:
\be
 \psi = \psi_+ + \psi_- \qquad \qquad \epsilon \psi_{\pm} = \pm \psi_{\pm}
\ee
Next we define a pair of future-directed null normals to $C$
\be
 l^a = \frac{1}{\sqrt{2}} \left( e_0^a + e_2^a \right) \qquad n^a = \frac{1}{\sqrt{2}} \left( e_0^a - e_2^a \right) 
\ee
so $l^a$ points ``out of'' $\Sigma$ and $n^a$ points ``into'' $\Sigma$.
We extend $l$ and $n$ off of $C$ as tangents to affinely parametrised null geodesics from each point of $C$, generating affinely parametrised null geodesic congruences with tangents $l$ and $n$.
For each congruence, we let $\theta_{\rm out}$ ($\theta_{\rm in}$) be the expansion on $C$ of the null geodesics with tangent $l^a$ ($n^a$) on $C$
\begin{equation}
    \theta_{\rm out} = \sqrt{2}\nabla_al^a \qquad \qquad \theta_{\rm in} = \sqrt{2}\nabla_a n^a
\end{equation}

Changing $\Sigma$ with $C$ fixed induces a $SO(1,1)$ Lorentz boost on the basis vectors $\{e_0^a,e_2^a\}$ normal to $C$. 
This has the effect $l^a \rightarrow f l^a$, $n^a \rightarrow f^{-1} n^a$ for some positive function $f$ on $C$, which we shall, henceforth, refer to as a \textit{boost}.
An object $T$ is said to have {\it boost weight} $b$ iff under a boost it transforms as $T\mapsto f^{b}T$.
As an example, $\theta_{\text{in}}$ has boost weight $b=-1$, while $\psi_{\pm}$ have $b=\pm1/2$. We define a boost-covariant derivative as \cite{Durkee:2010xq}
\begin{equation}
    {\cal D} \equiv D_1 - b \chi \qquad \qquad \chi \equiv -e_1^a n^b \nabla_a l_b
\end{equation}
where $D_a$ is the (flat) Levi-Civita connection on $C$ (and $D_1 \equiv e_1^a D_a$). Under a boost $\chi \rightarrow \chi + f^{-1} e_1^b \nabla_b f$, which ensures that ${\cal D}$ transforms covariantly.
We can also write
\be
 \chi = K_{12}
\ee
where $K_{ij}$ ($i,j=1,2$) is the extrinsic curvature of $\Sigma$.
By explicit computation one can verify that ${\cal D}$ is not only boost-covariant, but is also a derivation on objects of definite boost-weight.
Additionally, it satisfies the product rule when acting on polynomial functions of boost-weighted scalars.
We also define a ``slashed'' version of this operator
\begin{equation}
    \Ds \equiv \Gamma^2 \Gamma^1 {\cal D}
\end{equation}
The operator $\slashed{\cal D}$ is self-adjoint in the sense that, for any spinors $\eta,\psi$ defined near $C$, with opposite boost weights,
\be
\label{selfadjoint}
 \int_{C} \eta^\dagger   \slashed{\cal D} \psi =  \int_{C}  (\slashed{\cal D}\eta)^\dagger   \psi
\ee
The starting point for our definitions of quasilocal energy and angular momentum is the Witten identity for the supercovariant derivatives \eqref{supercovdef}:
\be
\label{wittenid}
\int_{\partial \Sigma} \psi^\dagger \Gamma^{2} \Gamma^1  \nabla^\pm_1 \psi=
 \int_\Sigma \left[  (\nabla_i^\pm \psi)^\dagger \nabla_i^\pm \psi +  4\pi G T_{0\mu} \bar{\psi} \Gamma^\mu \psi  - (\Gamma^i  \nabla_i^\pm \psi)^\dagger \Gamma^j  \nabla_j^\pm \psi\right] 
\ee
where $\psi$ is an arbitrary spinor defined in a neighbourhood of $\Sigma$. The integrals are defined w.r.t. the induced volume forms on $\partial\Sigma$ and $\Sigma$ respectively. The vector $\bar{\psi} \Gamma^\mu \psi$ is causal (or zero) so the dominant energy condition implies that the term $ T_{0\mu} \bar{\psi} \Gamma^\mu \psi$ is non-negative. Hence the RHS of \eqref{wittenid} is non-negative if the spinor satisfies the Witten equation 
$\Gamma^i  \nabla_i^\pm \psi = 0$ in $\Sigma$. However, following \cite{ludvigsen1983momentum,Dougan:1991zz}, we will {\it not} impose this equation on $\psi$ directly. We mention it just to motivate the idea that the LHS of \eqref{wittenid} might enjoy some positivity property. This LHS can be written as a sum over connected components of $\partial \Sigma$. For a connected component $C$ we define
\be
 I_C^\pm[\psi] \equiv \int_{C} \psi^\dagger \Gamma^{2} \Gamma^1  \nabla^\pm_1 \psi
\ee
and we can write this in terms of the derivative $\slashed{\cal D}$ defined above:
\begin{equation}
\label{ICexpr}
    I_{C}^\pm[\psi] = \int_{C} \left[  \psi_+^\dagger\left(  \slashed{\cal D} \psi_-  + \frac{1}{2} \theta_{\rm in}  \psi_+ \mp \frac{i}{2\ell}\Gamma^{2} \psi_-\right)+ \psi_-^\dagger \left(  \slashed{\cal D} \psi_+  - \frac{1}{2} \theta_{\rm out}   \psi_-  \mp \frac{ i}{2\ell}\Gamma^{2} \psi_+\right)\right]
\end{equation}
Now we focus on a particular connected component $S$ of $\partial \Sigma$ which we will regard as the ``outer'' boundary of $\Sigma$, for which we wish to define quasilocal quantities. We will assume that $S$ satisfies the convexity condition that {\it the future-directed ingoing null geodesics normal to $S$ are strictly converging on $S$}:
\be
\label{convex}
\theta_{\rm in}<0 \;\; {\rm on} \;\; S
\ee
We will restrict attention to spinors $\psi$ satisfying the following equation on $S$:
\begin{equation}
\label{hol1}
   \frac{1}{2}(1+\epsilon) \nabla^\pm_1 \psi \equiv     \slashed{\cal D} \psi_-  + \frac{1}{2} \theta_{\rm in}  \psi_+ \mp  \frac{i}{2\ell}\Gamma^{2} \psi_-=0 \;\; {\rm on} \;\; S
\end{equation}
The sign choice $\pm$ is held fixed: we are imposing only one equation on $\psi$. Since $\theta_{\rm in} \ne 0$ this equation uniquely determines $\psi_+$ in terms of $\psi_-$ on $S$. This is a 3d analogue of an equation used in \cite{ludvigsen1983momentum,Dougan:1991zz}. Following \cite{ludvigsen1983momentum,Dougan:1991zz}, we will show below, using \eqref{wittenid}, that if $\psi$ satisfies this equation and $\partial \Sigma$ satisfies certain conditions (e.g. if $\partial \Sigma = S$) then $I_S^\pm[\psi]$ is non-negative. 
Consider the following inner product on spinors with $b=-1/2$
\be
\label{eq:boost inv ip}
 \langle \eta,\phi \rangle = \int_S \frac{2}{-\theta_{\rm in}} \eta^\dagger \phi
\ee
Since $\theta_{\rm in}$ has $b=-1$, this inner product is boost invariant.
Substituting \eqref{hol1} into \eqref{ICexpr} gives
\be
\label{IO}
 I^\pm_S[\psi] = \langle\psi_-, {\cal O}^\pm \psi_-\rangle
\ee
where we have defined the following operators on $S$ that act on spinors of definite boost weight
\begin{equation}
\label{Odef}
    {\cal O}^{\pm}  = \left({\cal D} \pm \frac{i}{2\ell}\Gamma^1\right)^\ddagger \left({\cal D}\pm\frac{i}{2\ell}\Gamma^1\right) + \frac{\theta_{\rm in}\theta_{\rm out}}{4}
\end{equation}
Here, $\ddagger$ represents the formal adjoint in the inner product \eqref{eq:boost inv ip}, which is
\begin{equation}
    \left({\cal D}\pm \frac{i}{2\ell}\Gamma^1\right)^\ddagger \equiv -{\cal D} + \frac{1}{\theta_{\rm in}}\left({\cal D}\theta_{\rm in}\right) \pm \frac{i}{2\ell} \Gamma^1
\end{equation}
These are elliptic operators that are covariant under $SO(1,1)$ and hence their spectra are $SO(1,1)$ invariants and thus intrinsic to $S$.
These operators are self-adjoint in the boost-invariant inner-product \eqref{eq:boost inv ip} from which it follows that the eigenvalues are real.

Furthermore ${\cal O}^\pm$ commutes with $\epsilon$ so it preserves the property $\epsilon=-1$. An eigenfunction must be either periodic or anti-periodic around $S$. We refer to these as ``Ramond (R) sector'' and ``Neveu-Schwarz (NS) sector'' eigenfunctions respectively. 
We define $\Lambda^{\rm R}_\pm$ ($\Lambda^{\rm NS}_\pm$) to be the {\it lowest}\footnote{
Note that \eqref{Odef} shows that $\inf_S (\theta_{\rm in} \theta_{\rm out}/4)$ is a lower bound for the eigenvalues.} eigenvalue living in the R (NS) sector respectively:
\bea
 {\cal O}^\pm \psi_- &=& \Lambda^{\rm R}_\pm \psi_- \qquad \qquad  \psi_- \; {\rm periodic} \nonumber \\
{\cal O}^\pm \psi_- &=& \Lambda^{\rm NS}_\pm \psi_-  \qquad \qquad \psi_- \; {\rm antiperiodic}
\eea
Now we define the {\it non-extremalities} in each sector as
\be
\label{nonextdef}
\Delta^{\rm R}_\pm = r^2 \Lambda^{\rm R}_\pm \qquad \qquad  \Delta^{\rm NS}_\pm = r^2 \Lambda^{\rm NS}_\pm
\ee
where $r$ is the ``circumference-radius'' of $S$, i.e., the proper length around $S$ is $2\pi r$. There are two non-extremalities for each sector, corresponding to the different choices for the label $\pm$. This is familiar for asymptotic charges in AdS$_3$, where they describe ``left-moving'' and ``right-moving'' excitations in each sector. Here we have recovered this structure in a quasi-local setting. Finally our quasi-local energy and angular momentum are defined for each sector as
\be
 E^{\rm R/NS} = \frac{1}{4G} \left(\Delta^{\rm R/NS}_+ + \Delta^{\rm R/NS}_- \right)-\frac{\delta^{\rm R/NS}}{8G} \qquad J^{\rm R/NS} = \frac{\ell}{4G} \left(\Delta^{\rm R/NS}_+ - \Delta^{\rm R/NS}_- \right)
\ee
where $\delta^{\rm R}=0$ and $\delta^{\rm NS} = 1$.

Our definitions required the convexity condition \eqref{convex}, i.e., that the ingoing null geodesics from $S$ are strictly converging. Just as in \cite{Dougan:1991zz}, one can formulate analogous definitions if, instead of \eqref{convex}, we impose $\theta_{\rm out}>0$, i.e., that the outgoing null geodesics from $S$ are strictly expanding. This is discussed in Appendix \ref{sec:alternate_defs}.

We have defined {\it two} pairs of quasilocal quantities, corresponding to the labels R and NS. However, in axisymmetry these definitions coincide:
\begin{lemma}
\label{lem:axisym}
 Assume that $S$ coincides with a spacelike orbit of the rotational Killing vector field in an axisymmetric spacetime. Then $E^\rR=E^{\rNS}$ and $J^\rR=J^{\rNS}=J_{\rm Komar}$ where $J_{\rm Komar}$ is the Komar angular momentum of $S$. For such $S$ in the BTZ spacetime, these quantities coincide with the usual BTZ mass and angular momentum parameters.
\end{lemma}

We will prove this Lemma in section \ref{sec:lemproof} below. In section \ref{sec:asymptotic charges} we compute our charges in the limit where $S$ approaches a cross-section of conformal infinity in an asymptotically AdS$_3$ spacetime. In this case, the result coincides with standard definitions \cite{Brown:1986nw} of energy and angular momentum in an axisymmetric situation although not more generally. 

The main reason these definitions are interesting is supplied by the following theorem, which establishes that they satisfy a BPS inequality under suitable circumstances.

\begin{theorem}
\label{thm:bps}
    Let $\Sigma$ be a smooth, compact, connected, spacelike hypersurface with boundary $\partial \Sigma = S \cup T$ where $S$ (the ``outer'' boundary of $\Sigma$) is topologically $S^1$ and $T$ (the ``inner'' boundary of $\Sigma$) is a (possibly empty) finite union of weakly future outer trapped surfaces ($\theta_{\rm in} \le 0$ on $T$\footnote{
    Recall that ``in'' means ``directed into $\Sigma$.''}). Assume that the future-directed ingoing null geodesics normal to $S$ are strictly converging ($\theta_{\rm in} < 0$ on $S$). Assume that the dominant energy condition holds on $\Sigma$ and that $\Sigma$ admits a spin structure compatible with periodic (anti-periodic) spinors on $S$. Then the following is true for the R (NS) sector definitions of $E$ and $J$:\footnote{To be more explicit: if there exists a spin structure compatible with periodic spinors on $S$ then $\Delta_+$ means $\Delta_+^\rR$, $E$ means $E^\rR$ etc.}
    
    (a) $\Delta_+ \ge 0$ and $\Delta_- \ge 0$ hence $E \ge |J|/\ell-\delta/(8G)$ (the BPS inequality).

    (b) $\Delta_\pm =0$ (for one sign choice $\pm$) if, and only if, (i) there exists a non-trivial spinor $\psi$ on $\Sigma$, periodic (anti-periodic) around $S$, and satisfying $h^b_a\nabla^\pm_b \psi =0$, where $h^b_a$ is the projection onto $\Sigma$, (ii) $T$ (if non-empty) is marginally future outer trapped ($\theta_{\rm in} \equiv 0$ on $T$). (iii) at points of $\Sigma$ where $T_{ab} \ne 0$, $X^a \equiv \bar{\psi} \Gamma^a \psi$ must be null with $T_{ab} = \rho X_a X_b$ for some $\rho>0$. 
\end{theorem}

This theorem will be proved in section \ref{sec:proof of bps}. As an application of this theorem, take $S$ to be a circle enclosing the horizon of an extremal BTZ black hole. If $S$ is a spacelike circle of axisymmetry then Lemma \ref{lem:axisym} implies that the quasilocal quantities on $S$ agree with the usual extremal BTZ parameters, which are known to saturate the R-sector BPS bound \cite{Coussaert:1993jp}. If $S$ is {\it not} a circle of axisymmetry then we can use Theorem \ref{thm:bps} to establish the same result:

\begin{corollary}
\label{cor:btz}
 In an extremal BTZ black hole spacetime, consider a smooth spacelike hypersurface $\Sigma$ with the topology of an annulus with outer boundary $S$ and inner boundary $T$ where $T$ is a cross-section of the future event horizon. Assume that the future-directed ingoing null geodesics normal to $S$ are strictly converging. Then the R-sector quasi-local quantities defined on $S$ saturate the R-sector BPS bound.  
\end{corollary}
{\it Proof.} 
An extremal BTZ black hole spacetime admits a supercovariantly constant spinor satisfying ${\nabla}^\pm_a \psi =0$ where the choice of sign $\pm$ is determined by the direction of rotation and the spinor is periodic around $S$ \cite{Coussaert:1993jp}. $\Sigma$ admits a spin structure compatible with periodic spinors on $S$. No matter is present so the dominant energy condition is trivially satisfied. Hence the assumptions of Theorem \ref{thm:bps} are satisfied, and the conditions in (b) are all met. The ``if'' part of (b) implies that $\Delta_\pm^\rR=0$ so the R-sector BPS inequality is saturated for any such $S$. (Note that we have not shown that the R-sector quasilocal quantities must agree with the BTZ parameters except in axisymmetry.)

If $S$ is a marginally outer trapped surface, i.e., $\theta_{\rm out} \equiv 0$ then we can prove BPS inequalities without any further assumptions (and no energy condition):

\begin{lemma}
\label{lem:marg_trap}
Let $S$ be a marginally outer trapped surface for which the future-directed ingoing null geodesics normal to $S$ are strictly converging. 
Then the BPS inequalities $\Delta_\pm \ge 0$ hold in both the R and NS sectors.
Moreover, equality can only occur in the Ramond sector, and does so if, and only if, $\ker \left({\cal D} \pm \frac{i}{2\ell}\Gamma^1\right)$ is non-trivial on $S$. 
In particular, $\Delta_\pm^{\rm R}=0$ for a cross-section of the future event horizon of an extremal BTZ black hole (with $\theta_{\rm in}<0$), with the choice of sign $\pm$ determined by the sense of rotation of the black hole.
\end{lemma}
{\it Proof.}
Since $\theta_{\rm out} \equiv 0$, the first part of this Lemma follows immediately from the definition \eqref{Odef}. (The assumption $\theta_{\rm in}<0$ ensures that $\ddagger$ is well-defined in \eqref{Odef}.) It is also clear from \eqref{Odef} that the inequality is saturated iff ${\cal D} \pm \frac{i}{2\ell}\Gamma^1$ has a non-trivial kernel. Let $\psi$ be an element of this kernel. The operator is real and diagonal so each component of $\psi$ satisfies a 
real, linear, first order, homogeneous ODE. The real and imaginary parts of each component satisfy the same ODE. If $\psi$ is antiperiodic then these real and imaginary parts must each vanish somewhere so the ODE implies they vanish everywhere. Hence the kernel is trivial in the Neveu-Schwarz sector. For extremal BTZ, we note that a cross-section of the future event horizon is a marginally outer trapped surface and use the fact that the supercovariantly constant spinor of extremal BTZ restricts  to a periodic spinor on $S$ in the kernel of ${\cal D} \pm \frac{i}{2\ell}\Gamma^1$ (for one choice of sign $\pm$). This can be seen explicitly from equation \eqref{antihol} of Appendix \ref{sec:alternate_defs} (multiplied by $\Gamma^1 \Gamma^2$).

\subsection{Proof of the Third Law}

\label{sec:3rdlawproof}

We will now prove Theorem \ref{thm:3rdlaw}. We start by showing that, under the assumptions of Theorem \ref{thm:3rdlaw}, the R-sector quasilocal quantities on $S$ saturate the R-sector BPS bound. Our definitions of quasilocal quantities depend only on the spacetime metric on $S$ (not just the induced metric on $S$) and extrinsic curvature of $S$. Hence the assumption that, at $S$, the spacetime metric and extrinsic curvature are the same as for a horizon cross-section $S_{\rm BTZ}$ of extremal BTZ implies that $S$ has the same quasilocal quantities as $S_{\rm BTZ}$. From Lemma \ref{lem:marg_trap} we know that $S_{\rm BTZ}$ saturates the R-sector BPS bound (for one choice of sign $\pm$) and hence $S$ must also saturate this bound. 

Now observe that $\Sigma$ in the statement of Theorem \ref{thm:3rdlaw} satisfies the assumptions of Theorem \ref{thm:bps}, choosing the spin structure for which spinors are periodic around $S$. We have just shown that the R-sector BPS inequality is saturated so point (b) of Theorem \ref{thm:bps} implies that $T$ must be marginally future outer trapped. So we have proved Theorem \ref{thm:3rdlaw} except in the special case where $T$ is marginally future outer trapped. In particular we have proved it if $T$ is strictly future outer trapped ($\theta_{\rm in}<0$) which is all we really need to justify calling the result a Third Law.

More work is required to exclude the case where $T$ is marginally future outer trapped, so here we will give a sketch of the argument, which is similar to \cite{Reall:2024njy}. Observe that Theorem \ref{thm:bps} holds true for any such $\Sigma$ with $\partial\Sigma=S\cup T$. Since $D(\Sigma)$ (the domain of dependence of $\Sigma$) is foliated by such surfaces, one would expect that point (b) of the theorem implies that $D(\Sigma)$ carries a non-trivial spinor, $\psi$, parallel with respect to the appropriate $\hat\nabla^{\pm}$.
Indeed, in forthcoming work we will provide a proof of this result.
Given this, one can construct the (causal) vector $X^\mu \equiv \bar\psi\Gamma^\mu\psi$ in $D(\Sigma)$.
This is Killing as
\begin{equation}
    \nabla_\mu X_\nu = \pm \frac{i}{\ell}\bar\psi \Gamma_{\mu\nu}\psi\,.
\end{equation}
Since $\psi$ satisfies \eqref{hol1}, equation \eqref{ICexpr} evaluated on $S$ becomes, using the self-adjointness of $\Ds$ in \eqref{selfadjoint}
\begin{equation}
    I_{S}^{\pm}[\psi]=-\frac{1}{2}\int_S \left( \theta_{\rm in} \psi_+^\dagger \psi_+ + \theta_{\rm out}\psi_-^\dagger \psi_- \right) = -\frac{1}{2}\int_S\theta_{\rm in} \psi_+^\dagger \psi_+
\end{equation}
where the second equality follows because $S$ has the same extrinsic curvature as a horizon cross-section of a BTZ black hole and hence is marginally future outer trapped. However, because $\psi$ is $\nabla^{\pm}$-parallel, the LHS of the above equation vanishes. Since $\theta_{\rm in}<0$, we get $\psi_+^\dagger \psi_+= 0$, i.e. $\epsilon\psi = -\psi$ on $S$.
A similar argument gives $\epsilon\psi=\psi$ on $T$. From this, we compute
\begin{equation}
    (\star {\rm d}X)_1 = -({\rm d}X)_{02} = \mp\frac{2{\rm i}}{\ell}\bar\psi \epsilon\psi =0
\end{equation}
 the final equality holding on both $S$ and $T$ as $\psi$ is an eigenspinor of $\epsilon$ on both surfaces.
Therefore, using Stokes' theorem and a standard identity on Killing vectors, one has
\begin{equation}
    \int_{\partial \Sigma} \star {\rm d}X = 0 = \int_\Sigma {\rm d}\star {\rm d}X = \int_{\Sigma} R_{ab}n^a X^b
\end{equation}
where $n^a$ is the unit timelike normal to $\Sigma$.
Point b iii). of Theorem \ref{thm:bps} establishes that $R_{ab}X^b= 2\Lambda X_a$ in $D(\Sigma)$ as any matter present is either trivial or null dust with velocity vector $X$ (when $X$ itself is null).
The above becomes
\begin{equation}
    \int_{\Sigma} \Lambda g_{ab}n^a X^b = 0
\end{equation}
Since $X$ is everywhere causal, $n^aX_a$ is pointwise negative, and hence we have a contradiction with the existence of such a $T$.

We stated theorem \ref{thm:3rdlaw} for topologically annular $\Sigma$ for reasons of simplicity. However, the proof clearly generalises immediately to allow for many different topologies for $\Sigma$. Specifically it works if $\Sigma$ is a connected compact spacelike surface with $\partial \Sigma = S\cup T \cup A$ where $S$ is topologically $S^1$, $T$ is a finite union of weakly future outer trapped surfaces, $A$ is a finite union of weakly future inner anti trapped surfaces and $\Sigma$ admits a spin structure compatible with periodic spinors on $S$. Importantly the latter property excludes the case for which  $\Sigma$ has the topology of a disc and $T$ is the empty set, which describes a black hole formed in gravitational collapse.

\subsection{Proof of Theorem \ref{thm:bps}}\label{sec:proof of bps}

We will actually prove a slight generalisation of \ref{thm:bps} taking $\partial \Sigma = S\cup T \cup A$ where $A$ is a finite union of weakly future inner antitrapped surfaces, meaning $\theta_{\rm out} \ge 0$ on $A$\footnote{We can, in fact, prove a similar theorem with the weaker assumption that, on $S$, $\theta_{\rm in}\leq 0$ except for at least one point where $\theta_{\rm in}<0$.
In this setting, however, our quasi-local quantities are not defined and so the theorem amounts to proving that $I_S^\pm[\psi]$ itself is a non-negative functional, with a similar rigidity associated to the case of it vanishing.}
. In part (b)(ii) of the theorem we include the statement that $\theta_{\rm out} \equiv 0$ on $A$ (if non-empty).

To prove (a), our strategy is to show that $I_S^\pm[\psi] \ge 0$ whenever \eqref{hol1} is satisfied. If we can do this then \eqref{IO} implies that the eigenvalues of ${\cal O}^\pm$ are non-negative, so $\Delta_\pm \ge 0$, establishing (a). 

We follow \cite{ludvigsen1983momentum,Dougan:1991zz}, slightly generalized to allow for an inner boundary. Let $\psi$ be a spinor satisfying \eqref{hol1} for which $\psi_-$ is periodic (antiperiodic) around $S$. Equation \eqref{hol1} then implies that $\psi_+$ is also periodic (antiperiodic) around $S$ and so $\psi$ is  periodic (antiperiodic) around $S$. Let $\tilde{\psi}$ be another periodic (antiperiodic) non-trivial spinor with $\tilde{\psi}_- = \psi_-$ on $S$. A straightforward argument using \eqref{selfadjoint} gives
\be
\label{psitildepsi}
 I_{S}^\pm[\psi]= I_{S}^\pm[\tilde{\psi}] +   \int_{S} \left(-\frac{1}{2} \theta_{\rm in}\right) |\tilde{\psi}_+ - {\psi}_+|^2
\ee
We will take $\tilde{\psi}$ to be a solution of the Witten equation  on $\Sigma$ (for a fixed choice of sign $\pm$)
\be
\label{witteneq}
 \Gamma^i \nabla^\pm_i \tilde{\psi}=0
\ee 
with boundary conditions 
\be
\label{wittenbcs}
\tilde{\psi}_- = \psi_- \;\; {\rm on} \;\; S \qquad \qquad \tilde{\psi}_-=0 \;\; {\rm on} \;\;  T \qquad \qquad  \tilde{\psi}_+=0\;\; {\rm on} \;\;  A 
\ee
It is here that we need the assumption that $\Sigma$ admits a spin structure compatible with the periodicity (anti-periodicity) of $\psi$, for otherwise no such spinor will exist for non-vanishing $\psi_-$. Existence and uniqueness of $\tilde{\psi}$ follow from standard arguments as follows. 

To establish uniqueness of $\tilde{\psi}$ we apply \eqref{wittenid} to the difference $\Delta \tilde{\psi}$ of two such solutions: writing $I^\pm_{\partial \Sigma} = I^\pm_S + I^\pm_T + I^\pm_A$, expressing each of these terms as in \eqref{ICexpr} and using the boundary conditions \eqref{wittenbcs} to simplify gives
\bea
\label{unique}
\int_S \frac{1}{2} \theta _{\rm in} (\Delta \tilde{\psi})_+^\dagger (\Delta \tilde{\psi})_+ &=& \int_T \left(-\frac{1}{2} \theta _{\rm in}\right) (\Delta \tilde{\psi})_+^\dagger (\Delta \tilde{\psi})_+ + \int_A \frac{1}{2} \theta _{\rm out} (\Delta \tilde{\psi})_-^\dagger (\Delta \tilde{\psi})_- \nonumber \\ &+& \int_\Sigma \left[(\nabla_i^\pm (\Delta \tilde{\psi}))^\dagger \nabla_i^\pm (\Delta \tilde{\psi}) +4\pi G T_{0\alpha} \Delta \bar{ \tilde {\psi}} \Gamma^\alpha \Delta\tilde{\psi} \right]
\eea
We have $\theta_{\rm in} < 0$ on $S$, $\theta_{\rm in} \le 0$ on $T$ and $\theta_{\rm out} \ge 0$ on $A$. Hence the LHS is non-positive and the RHS is non-negative (using the dominant energy condition and the fact that $\Delta \bar{ \tilde {\psi}} \Gamma^\alpha \Delta\tilde{\psi}$ is a causal vector or zero) and so both must vanish. Vanishing of the LHS implies that $(\Delta \tilde{\psi})_+=0$ on $S$ and hence (from \eqref{wittenbcs}) $\Delta \tilde{\psi}=0$ on $S$. Vanishing of the RHS implies that $\nabla^\pm_i \Delta \tilde{\psi}$ vanishes on $\Sigma$. Combining these facts and using path-connectedness of $\Sigma$, implies that $\Delta \tilde{\psi}$ vanishes everywhere on $\Sigma$, establishing uniqueness. 

To establish existence of $\tilde{\psi}$ write $\tilde{\psi} = \psi' + \psi_0$ where $\psi_0$ is a smooth spinor obeying $\psi_{0-}=\psi_-$ on $S$ and $\psi_0=0$ on $T,A$. We then have to solve $\Gamma^i \nabla^\pm_i \psi'=f$ where $f = -\Gamma^i \nabla^\pm_i \psi_0$ with boundary conditions $\psi'_-=0$ on $S,T$ and $\psi'_+ = 0$ on $A$. To do this, choose $\psi'$ to minimize $\int_\Sigma |\Gamma^i \nabla^\pm_i \psi'-f|^2$ in the space of smooth spinors on $\Sigma$ satisfying these boundary conditions. Here we invoke the assumption that $\Sigma$ admits a spin structure compatible with the periodicity (antiperiodicity) of $\psi$ for otherwise this space is trivial. Taking the first variation of this expression and integrating by parts shows that the minimizing function satisfies
\be
 0 = \int_\Sigma \left[ (\Gamma^i \nabla^\mp_i \eta)^\dagger \delta \psi' + \delta \psi'^\dagger \Gamma^i \nabla^\mp_i\eta \right] + \int_{\partial \Sigma} \left[ \eta^\dagger \Gamma^2 \delta \psi' + \delta \psi'^\dagger \Gamma^{2} \eta \right]
\ee
where $\eta = \Gamma^i \nabla^\pm_i \psi'-f$. This must hold for arbitrary $\delta \psi'$ satisfying the boundary conditions. Hence $\Gamma^i \nabla^\mp_i \eta =0$ on $\Sigma$ and $\eta^\dagger \Gamma^{2} \delta \psi'=0$ on $\partial \Sigma$. Evaluating the latter condition on $S,T$ gives $0 = \eta_-^\dagger \Gamma^2 \delta \psi'_+$ and hence $\eta_-=0$ on $S,T$. Similarly $\eta_+=0$ on $A$. Hence $\eta$ satisfies exactly the same conditions as $\Delta \tilde{\psi}$ of the uniqueness argument and so the same argument shows that $\eta \equiv 0$, which establishes existence.

Having constructed $\tilde{\psi}$, the identity \eqref{witteneq} shows that $I_S^\pm[\tilde{\psi}]$ is given by replacing $\Delta \tilde{\psi}$ with $\tilde{\psi}$ on the RHS of \eqref{unique}:
\be
\label{identity2}
I_S^\pm[\tilde{\psi}] =  \int_T \left(-\frac{1}{2} \theta _{\rm in} \right)  \tilde{\psi}_+^\dagger \tilde{\psi}_+ + \int_A \frac{1}{2} \theta _{\rm out}  \tilde{\psi}_-^\dagger \tilde{\psi}_- + \int_\Sigma \left[(\nabla_i^\pm  \tilde{\psi})^\dagger \nabla_i^\pm \tilde{\psi} +4\pi G  T_{0\alpha} \bar{ \tilde {\psi}} \Gamma^\alpha \tilde{\psi} \right] 
\ee
which is manifestly non-negative if the dominant energy condition is satisfied on $\Sigma$. Finally from \eqref{psitildepsi} we deduce that $I_S^\pm[\psi]\ge 0$ for any $\psi$ satisfying \eqref{hol1}. This completes the proof of (a).

For (b), if $\Delta_\pm=0$ then the lowest eigenvalue of ${\cal O}^\pm$ is zero, i.e., ${\cal O}^\pm$ has non-trivial kernel in the periodic (antiperiodic) sector. Take $\psi_-$ to be a non-zero element of this kernel, with $\psi_+$ determined by \eqref{hol1}. We then have $I^\pm_S[\psi]=0$ from \eqref{IO}. Equation \eqref{psitildepsi} implies that $I^\pm[\tilde{\psi}]=0$ and that $\tilde{\psi} = \psi$ on $S$ so $\tilde{\psi}$ is an extension of $\psi$ onto $\Sigma$ and so we can drop the tilde. \eqref{identity2} then implies that $\psi$ must satisfy $\nabla_i^\pm \psi=0$ on $\Sigma$, proving (i). Furthermore if there is a point on $T$ at which $\theta_{\rm in}<0$ or a point on $A$ at which $\theta_{\rm out}>0$ then $\psi$ must vanish at that point, which then implies (using $\nabla_i^\pm \psi=0$) that $\psi$ vanishes on $\Sigma$ and hence $\psi$ vanishes on $S$, a contradiction. So we must have $\theta_{\rm in} \equiv 0$ on $T$ (if non-empty) and $\theta_{\rm out}\equiv 0$ on $A$ (if non-empty), proving (ii). Finally we have $T_{0\alpha} X^\alpha=0$ on $\Sigma$ where $X^a = \bar{\psi} \Gamma^a \psi$. Hence $T_{ab} e_0^a X^b=0$. But $e_0^a$ is timelike and $X^b$ is causal so the dominant energy condition gives $T_{ab} X^b=0$. Hence $T_{ab} Y^a X^b=0$ for any timelike vector $Y^a$. If $X^a$ is timelike this implies $T_{ab} Y^b=0$ and hence $T_{ab}=0$. If $X^a$ is null then it implies $T_{ab} Y^b \propto X_a$ hence $T_{ab} = X_a \omega_b$ for some $\omega_b$ and antisymmetrizing gives $\omega_b = \rho X_b$ for some $\rho$. Finally $\rho \ge 0$ from the dominant energy condition.

Conversely, if (b)(i), (ii), (iii) hold (for some choice of sign $\pm$) then \eqref{identity2} (without tildes) implies that $I_S^\pm[\psi]=0$. Such $\psi$ cannot vanish on $S$ for otherwise (using $\nabla_i^\pm \psi=0$) it would be identically zero on $\Sigma$. Furthermore $\psi_-$ is non-trivial on $S$ for otherwise \eqref{convex} and \eqref{hol1} would imply $\psi=0$ on $S$. From (a) we know that the eigenvalues of ${\cal O}^\pm$ are non-negative. Hence \eqref{IO} implies that $\psi_-$ is a non-trivial element of the kernel of ${\cal O}^\pm$ so $\Lambda_\pm=0$ and hence $\Delta_\pm=0$. This completes the proof.

{\it Remarks.} We conclude this section with some remarks on the assumption that $\Sigma$ admits a spin structure compatible with the one chosen on $S$. Let $\Sigma$ have genus $g$ and $b$ boundary components. Its spin structures are in one-to-one correspondence with the elements of $H^{1}(\Sigma;\mathbb{Z}_2)$ \cite[Ch.~II, Thm.~1.7]{LawsonMichelsohn}, and hence there are $2^{2g+b-1}$ such structures. The induced spin structures on the boundary components are constrained by the requirement that the number of components carrying the Ramond (periodic) spin structure be even.

Consequently, if $T$ and $A$ are empty, so that $S$ is the only boundary component, the spin structure induced on $S$ must be Neveu-Schwarz. Thus in this case Theorem~2.2 applies only to the NS-sector quantities, independently of the genus of $\Sigma$. If $T$ or $A$ is non-empty, then for either choice of spin structure on $S$ one can choose the spin structures on the components of $T$ or $A$ so that this condition is satisfied. Hence Theorem~2.2 can be applied to both the Ramond and Neveu-Schwarz sectors whenever $T$ or $A$ is non-empty.

\subsection{Quasilocal energy and angular momentum in axisymmetry}

\label{sec:lemproof}

\subsubsection{Proof of Lemma \ref{lem:axisym}}

The proof of lemma \ref{lem:axisym} amounts to evaluating the operators ${\cal O}^{\pm}$ in axisymmetry, and relating the point spectra of these operators on an axisymmetric circle to the spectrum of the Laplacian intrinsic to the circle.
Since the spectrum of the intrinsic Laplacian is non-negative, the minimum of this determines the minima of the spectra of ${\cal O}^{\pm}$.

Let $m^a$ be the Killing vector field associated with axisymmetry.
We can choose $l$ and $n$ to be invariant w.r.t. $m$: ${\cal L}_ml={\cal L}_m n =0$.
This implies that both $\theta_{\text{out}}$ and $\theta_{\text{in}}$ are constants.
Additionally, one can show $\chi$ is a constant function on $S$ in this gauge.
We then obtain
\begin{equation}\label{eq:operators in axisymmetry}
    {\cal O}^{\pm} = \left(\Gamma^1 D_1\right)^2 + \frac{\theta_{\text{out}}\theta_{\text{in}}}{4} + \frac{1}{4}\left( \chi\mp\frac{1}{\ell}\right)^2
\end{equation}
The lowest eigenvalues of ${\cal O}^{\pm}$ are thus determined by the lowest eigenvalue of the operator $\left(\Gamma^1 D_1\right)^2$ - the square of (two copies of) the intrinsic Dirac operator to $S$\footnote{by the classical Schr\"odinger-Lichnerowicz identity, this is, equivalently, just (two copies of) the intrinsic Laplacian - up to a sign - as there is no scalar curvature.}.
As is well-known - see \cite{ginoux2009dirac}, for example - the spectrum of the Dirac operator on $S$ is given by 
\begin{equation}
    \text{spec}(\Gamma^1  D_1) =\frac{1}{r}\left(\mathbb{Z}+\frac{\delta^{\rm R/NS}}{2}\right).
\end{equation}
where $r$ is the circumference-radius of $S$.
Therefore, the minima of the spectra of ${\cal O}^{\pm}$ are
\begin{equation}\label{eq:e-values in axisymmetry}
    \Lambda_\pm^{{\rm R/NS}} = \frac{\delta^{\rm R/NS}}{4r^2} +\frac{\theta_{\text{out}}\theta_{\text{in}}}{4} + \frac{1}{4}\left( \chi\mp\frac{1}{\ell}\right)^2
\end{equation}
We now make two definitions, valid only in axisymmetry.
Firstly, the \textit{Komar angular momentum} of a circle of axisymmetry $S$ is
\begin{equation}
    J_{\rm Komar} \equiv \frac{1}{16\pi G}\int_S \star {\rm d}m = \frac{1}{8G}i_m \star {\rm d}m,
\end{equation}
where the volume form is $e^0 \wedge e^1 \wedge e^2$. In the gauge chosen above, $\chi$ can be seen to be simply related to the Komar angular momentum. Since $e_1^a = r^{-1} m^a$, the invariance of $l$ w.r.t. $m$ implies
\begin{equation}
    \chi = -e_1^a n^b \nabla_a l_b = -\frac{1}{r} l^{[a} n^{b]} \nabla_a m_b
\end{equation}
since $m$ is Killing.
We then use that
\begin{equation}
    l^{[a}n^{b]} = \frac{1}{2r}m^c\epsilon{_c}^{ab}
\end{equation}
which follows directly from the definitions of $n$ and $l$ and where $\epsilon$ denotes the volume form, to obtain
\begin{equation}
    \chi = -\frac{1}{2r^2} i_m \star {\rm d}m = \frac{-4GJ_{{\rm Komar}}}{r^2}
\end{equation}
Secondly, we define the \textit{renormalised Hawking mass} of $S$ as
\begin{equation}\label{eq:renormed hawking mass}
    \varpi \equiv -\frac{1}{8G}g^{-1}({\rm d}r,{\rm d}r) + \frac{r^2}{8G\ell^2}+ \frac{2GJ_{\rm Komar}^2}{r^2}=\frac{r^2}{8G} \theta_{\text{out}}\theta_{\text{in}} + \frac{r^2}{8G\ell^2}+ \frac{2GJ_{\rm Komar}^2}{r^2}
\end{equation}
This second equality can be seen by noting that $\theta_{\rm out}=\sqrt{2}l^a\nabla_a(\ln r)$ and likewise for $\theta_{\rm in}$.
So,
\begin{equation}
    \theta_{\rm in} \theta_{\rm out} = \frac{2}{r^2}l^a n^b ({\rm d}r)_a ({\rm d}r)_b = -\frac{1}{r^2}g^{-1}({\rm d}r,{\rm d}r)
\end{equation}
the second equality following as ${\rm d}r$ is normal to $S$.

Using these definitions in equation \eqref{eq:e-values in axisymmetry}, the non-extremalities associated with each sector become
\begin{equation}
    \Delta_{\pm}^{{\rm R/NS}} = 2G\left(\frac{\delta^{\rm R/NS}}{8G} + \varpi\pm \frac{J_{\rm Komar}}{\ell}\right)
\end{equation}
Hence, taking sums and differences, we see that, on a circle of axisymmetry in an axisymmetric spacetime,
\begin{equation}
    E^{\rm R/NS} = \varpi,\qquad J^{\rm R/NS} = J_{\rm Komar}
\end{equation}
which establishes the first result of the lemma. To establish the second result, we now evaluate these quasi-local quantities in the BTZ metric.


\subsubsection{BTZ solution}

The metric of the BTZ solution is
\be\label{eq:btz metric}
{\rm d}s^2=-\frac{(r^2-r_+^2)(r^2-r_-^2)}{r^2\ell^2}{\rm d}t^2+\frac{\ell^2r^2{\rm d}r^2}{(r^2-r_+^2)(r^2-r_-^2)}+r^2\left({\rm d}\phi-\frac{\varepsilon r_+r_-}{r^2\ell}{\rm d}t\right)^2\,
\ee
where $r_+ \ge r_- \ge 0$ and $\varepsilon \in \{1,-1\}$. Computing the quasilocal energy and angular momentum of a circle of axisymmetry (constant $t,r$) gives
\be
\label{BTZMJ}
 \varpi = \frac{r_+^2 + r_-^2}{8G\ell^2} \qquad \qquad J_{\rm Komar} = \frac{\varepsilon r_+r_-}{4G\ell}
\ee
These expressions are independent of $r$ and agree with the usual BTZ mass and angular momentum parameters. The extremal BTZ solution has $r_+=r_->0$ and hence
\be
\label{eq:params at extremality}
 \varpi^{\rm EBTZ} = \frac{r_+^2}{4G\ell^2} \qquad \qquad J^{\rm EBTZ} = \frac{\varepsilon r_+^2}{4G\ell}
\ee

\subsection{Quasilocal charges evaluated at infinity}\label{sec:asymptotic charges}

In the following, we will investigate how our definitions of quasi-local energy and angular momentum behave when the circle $S$ is taken to infinity in a spacetime that is asymptotically AdS$_3$ in the sense of Brown and Henneaux \cite{Brown:1986nw}. The asymptotic symmetry algebra of such a spacetime is two copies of the Virasoro algebra, both with the same central extension defined in terms of the AdS radius \cite{Brown:1986nw}. In Fefferman-Graham gauge the spacetime metric takes the form \cite{Compere:2018aar,Compere:2013bya}
\bea \label{eq:aads3}
        g_{tt} &=& -\frac{r^2}{\ell^2} + S+O\left(r^{-2}\right), \qquad
        g_{t\varphi} = \ell \delta + O\left(r^{-2}\right), \qquad
        g_{\varphi\varphi} = r^2 + \ell^2 S + O\left(r^{-2}\right), \nonumber \\
        g_{rr} &\equiv & \frac{\ell^2}{r^2}, \qquad g_{rt} \equiv 0, \qquad g_{r\varphi} \equiv 0
\eea
where the conformal boundary is at $r \rightarrow \infty$. Such an expansion can always be done for any conformally compact metric in a neighbourhood of the conformal boundary \cite{Skenderis:2002wp}. We assume that any matter in the bulk of the spacetime falls off in a way that is compatible with the above behaviour of the metric. $S$ and $\delta$ are functions of $t,\varphi$ that take the form
\begin{equation}
    \delta\equiv L_{-}-L_+,\qquad S\equiv L_++L_-
\end{equation}
where $L_\pm$ are periodic functions of $x^{\pm}\equiv t/\ell \pm \varphi$.
The inverse metric is
\begin{equation}
    \begin{split}
        g^{-1} = \left(-\frac{\ell^2}{r^2}-\frac{\ell^4 S}{r^4} + O\left(r^{-6}\right)\right)\partial_t\otimes\partial_t + &2\left(\frac{\ell^3\delta}{r^4} + O\left(r^{-8}\right)\right)\partial_t\otimes\partial_\varphi + \\&\frac{r^2}{\ell^2}\partial_r\otimes\partial_r + \left(\frac{1}{r^2} - \frac{\ell^2 S}{r^4} + O\left(r^{-6}\right)\right)\partial_\varphi\otimes\partial_\varphi
    \end{split}
\end{equation}
Consider, now, a spacelike hypersurface of constant $t$.
This can be foliated by spacelike circles $C_r$ of constant $r$, upon each of which we define the operator ${\cal O}^{\pm}({r})\equiv {\cal O}^\pm_{C_r}$ according to \eqref{Odef}. We aim to determine the large $r$-behaviour of the eigenvalues of ${\cal O}^{\pm}({r})$. Our orthonormal basis must have $e_1^a$ tangent to $C_r$ so we choose $e_1^a = g_{\varphi\varphi}^{-1/2} (\partial/\partial \varphi)^a$. We then choose $e_0$ normal to the surface of constant $t$ containing $C_r$, i.e., $(e_0)_a = -(-g^{tt})^{-1/2} (dt)_a$. This then fixes $(e_2)^a = g_{rr}^{-1/2}(\partial/\partial r)^a$. Our basis is therefore
\begin{equation}
    \begin{split}
           e_0 &= -\left(\frac{r}{\ell} - \frac{\ell S}{2r} + O\left(r^{-3}\right)\right){\rm d}t, \qquad e_2 = \frac{\ell}{r}{\rm d}r ,\\
    \qquad e_1 &=\left(r+\frac{\ell^2 S}{2r} + O\left(r^{-3}\right)\right){\rm d}\varphi + \left(\frac{\ell\delta}{r} + O\left(r^{-3}\right)\right){\rm d}t 
    \end{split}
\end{equation}
In this basis, we obtain
\begin{subequations}
    \begin{equation}
        \theta_{\rm out} = \frac{1}{\ell} - \frac{\ell S}{ r^2} - \frac{\ell^2 \partial_\varphi\delta}{2 r^3}+O\left(r^{-4}\right)
    \end{equation}
    \begin{equation}\label{eq:ingoing expansion laurent}
        \theta_{\rm in} = -\frac{1}{\ell}  + \frac{\ell S}{r^2}  - \frac{\ell^2\partial_{\varphi}\delta}{2r^3}+O\left(r^{-4}\right)
    \end{equation}
    \begin{equation}
        \chi = \frac{\ell \delta}{r^2} -\frac{\ell^3\delta S}{r^4} + O\left(r^{-6}\right)
    \end{equation}
\end{subequations}
From this we obtain
\begin{equation}
    {\cal O}^{\pm}(r) = \frac{{\cal L}_{L_\pm}}{r^2}+ O(r^{-3})
\end{equation}
where we have introduced the \textit{Hill operators} \cite{magnus2004hill}
\be
 {\cal L}_{L_\pm}(\varphi) \equiv -\partial_\varphi^2 + L_{\pm}
\ee
In the limit $r\to\infty$, $C_r$ approaches a cross-section $C_\infty$ of the conformal boundary, parametrised by $\varphi$. The Hill operators can be regarded as acting on fields defined on $C_\infty$. For large $r$, the circumference-radius of $C_r$ is just $r$. Hence we can compute the large $r$ limit of the non-extremalities \eqref{nonextdef} associated with each sector (R or NS) as
\begin{equation}
\label{quasi_hill}
    \Delta_{\pm}^{\rm R/NS}(C_\infty) \equiv \lim_{r\rightarrow \infty} \Delta_{\pm}^{\rm R/NS}(C_r) =\mu^{\rm R/NS}[L_\pm]
\end{equation}
where $\mu^{\rm R/NS}[L_\pm]$ denote the lowest eigenvalue of the Hill operators in each sector.
We see that the non-extremalities in each sector are functionals of the undetermined parts of the metric \eqref{eq:aads3}.
Using the min-max theorem, we can express these as\footnote{From equation \eqref{eq:min-max}, these eigenvalues are bounded from below by $\delta^{R/NS}/4+ \min L_\pm$. 
}
\begin{equation}\label{eq:min-max}
    \mu^{\rm R/NS}[L_\pm] = \min_{\psi}\frac{\langle\psi,\left(-\partial_\varphi^2 + L_{\pm}\right)\psi\rangle_{C_\infty}}{\langle\psi,\psi\rangle_{C_\infty}}
\end{equation}
where we evaluate this in the boost-invariant inner product defined in \eqref{eq:boost inv ip}, with $\theta_{\rm in}$ taking its $r\to\infty$ value in \eqref{eq:ingoing expansion laurent} (in this gauge, the limit of the boost-invariant inner product is directly proportional to the standard $L^2$ inner product on such spinors).

Using the above result, our definitions give the energy and angular momentum of $C_\infty$ as
\begin{equation}
    \begin{split}
            E^{\rm R/NS}(C_\infty) &= \frac{1}{4G}\left(\mu^{\rm R/NS}[L_+]+\mu^{\rm R/NS}[L_-]\right) - \frac{\delta^{\rm R/NS}}{8G},\\ J^{\rm R/NS}(C_\infty) &= \frac{\ell}{4G}\left(\mu^{\rm R/NS}[L_+]-\mu^{\rm R/NS}[L_-]\right)
    \end{split}
\end{equation}
The standard definitions \cite{Brown:1986nw} of energy and angular momentum in AdS$_3$ are determined only by the zero modes (constant parts) of $L_\pm$. However, the eigenvalues of the Hill operators depend on the entirety of the functions $L_\pm$, not just on their zero modes and so our definitions are not determined solely by these zero modes. Hence, in general, the large $r$ limits of our quasilocal definitions do not agree with the standard definitions. However, when $L_{\pm}$ are constants - call them $\bar L_\pm$ - the situation is far simpler: one can easily compute the lowest eigenvalues to be
\begin{equation}\label{eq:min of hill}
    \Delta^{\rm R/NS}_{\pm} =\mu^{\rm R/NS} [\bar L_{\pm}] = \frac{\delta^{\rm R/NS}}{4} + \bar L_\pm
\end{equation}
Thus, the $E^{\rm R/NS}(C_\infty)$ and $J^{\rm R/NS}(C_\infty)$ reduce to
\begin{equation}
    E^{\rm R/NS}(C_\infty) = \frac{1}{4G}\left(\bar L_+ + \bar L_-\right),\qquad J^{\rm R/NS}(C_\infty) = \frac{\ell}{4G}\left(\bar L_+-\bar L_-\right)
\end{equation}
in agreement with the standard definitions \cite{Coussaert:1993jp, Compere:2018aar}. In particular this holds for the case of an axisymmetric spacetime for which the Fefferman-Graham gauge is adapted to axisymmetry.

Now we will discuss how our charges transform under the asymptotic symmetries of AdS$_3$. The asymptotic symmetry algebra consists of two copies of the (centrally extended) Virasoro algebra \cite{Brown:1986nw}. Asymptotic symmetries act on the conformal boundary as $x^\pm \mapsto h_\pm(x^\pm)$, which gives a conformal transformation of the boundary metric. This takes one out of the Fefferman-Graham gauge used above so a further coordinate transformation is required to bring the metric back to this gauge (see equation (2.26) of \cite{compere2016symplectic}). This results in a change of the functions $L_\pm$:
\begin{equation}\label{eq:Hill potential change}
    L_\pm(x^\pm)\mapsto \widetilde L_\pm \left(h_\pm(x^\pm)\right) \equiv (h_\pm'(x^\pm))^{-2}\left\{L_\pm\left(x^\pm\right) +\frac{1}{2}\left({\rm S}h_\pm\right)(x^\pm)\right\}
\end{equation}
where
\begin{equation}
    \left({\rm S}h\right)(x) \equiv \frac{h'''}{h'}-\frac{3}{2}\left(\frac{h''}{h'}\right)^2
\end{equation}
is the Schwarzian derivative of $h(x)$ w.r.t. $x$. The form of this transformation is dictated by the fact that $L_\pm$ are vectors in the coadjoint representation of the Virasoro algebra \cite{Oblak:2016eij}.

To investigate how our charges transform under an asymptotic symmetry we now restrict to asymptotic symmetries that preserve the cross-section $C_\infty$. Recall that this is a cross-section of constant $t$, say $t=0$, so invariance of $C_\infty$ requires that $h_+(\varphi)=-h_-(-\varphi)$, i.e., $C_\infty$ is preserved by a ``diagonal'' Virasoro subalgebra of our original asymptotic symmetry algebra. Using \eqref{eq:Hill potential change}, a tedious exercise shows that the Hill operators transform as \cite{balog1998coadjoint, Oblak:2016eij} 
\begin{equation}\label{eq:transformed Hill}
    {\cal L}_{L_+}[\cdot]\mapsto {\cal L}_{\widetilde L_+}[\cdot] = \left(h_+'\right)^{-\frac{3}{2}} {\cal L}_{L_+}\left[ \left(h_+'\right)^{-\frac{1}{2}}\cdot \right]
\end{equation}
and likewise for ${\cal L}_{L_-}$ with $h_+(\varphi)$ replaced by $-h_+(-\varphi)$. The form of this transformation implies that the eigenvalues of the Hill operators do {\it not} transform simply under asymptotic symmetries that preserve $C_\infty$.\footnote{Note that if $\psi$ is an eigenfunction of ${\cal L}_{L_+}$ with non-zero eigenvalue then $(h_+')^{1/2}\psi$ is {\it not} an eigenfunction of ${\cal L}_{\widetilde L_+}$.} Thus in general the asymptotic limits of our quasilocal charges do not transform simply. From a bulk perspective, the non-trivial nature of the transformation arises from the fact an asymptotic symmetry that preserves $C_\infty$ in general will not preserve the constant $t$ surface ending at $C_\infty$. Hence applying an asymptotic symmetry to the large $r$ limit of our quasilocal quantities corresponds to computing the limit to $C_\infty$ along a different sequence of finite circles $C_r$. 

There is an important case for which our charges do transform simply. This is the situation in which one of our non-extremalities vanishes (at infinity), corresponding to the preservation of some supersymmetry. From \eqref{quasi_hill} this corresponds to the existence of a zero-mode for one of the Hill operators.\footnote{A relation between zero modes of the Hill equation and supersymmetry was established in \cite{OColgain:2016msw}.} We observe from equation \eqref{eq:transformed Hill} that if a spinor $\psi$ is a zero-mode of ${\cal L}_{L_+}$ then $\tilde\psi \equiv \left(h_+'\right)^{1/2}\psi$ is a zero mode of ${\cal L}_{\widetilde L_+}$. Similarly for a zero-mode of ${\cal L}_{L_-}$ with the replacement of $h_+(\varphi)$ by $-h_+(-\varphi)$. It follows that the vanishing of, say, $\Delta^R_+[C_\infty]$ is a statement that is invariant under asymptotic symmetries that preserve $C_\infty$. 

This discussion also helps us address the question of whether our quasilocal energy is really an energy, rather than a mass. First, let us make this distinction clearer by way of analogy with flat space. In four-dimensional asymptotically flat spacetimes, the asymptotic symmetry group at null infinity is ${\rm BMS}_4$, an enhancement of the Poincar\'e group to include supertranslations. There is a canonical translation subgroup from which we construct the Bondi four-momentum. The Minkowski norm of this is the Bondi mass. The latter is invariant under the action of the BMS group, while the former is not \cite{Sachs:1962zza}. In an asymptotically AdS$_3$ spacetime we might similarly expect a definition of mass to be invariant under the asymptotic symmetry group whereas an energy would not. Indeed the recent work \cite{Rallabhandi:2025iyk} has provided a definition of quasilocal mass for 4d spacetimes with negative cosmological constant which has this property when evaluated at infinity. However, the above discussion demonstrates that the large $r$ limit of our quantities transform non-trivially under asymptotic symmetries. Furthermore, although they do not always agree with the standard definitions at infinity, the fact that they give rise to BPS inequalities of the standard form (Theorem \ref{thm:bps}) suggests that they are more energy-like than mass-like.

\section{Formation of extremal BTZ in gravitational collapse}

\label{sec:gluing}

\subsection{Equations of motion}

In this section, we consider a particular form of matter. Specifically, we take a simple example: a massless \emph{complex} scalar field $\Phi$ in AdS$_3$. The action is given by
\begin{equation}
\label{action}
S=\frac{1}{16\pi G}\int_{\mathcal{M}}{\rm d}^3x\sqrt{-g}\left[R+\frac{2}{\ell^2}-2\nabla_a\Phi \nabla^a \overline{\Phi}\right]\,,
\end{equation}
where the overline denotes complex conjugation. This can be obtained as a consistent truncation of type IIB supergravity with AdS$_3 \times S^3 \times \mathbb{T}^4$ asymptotics (see Section~8 of~\cite{Marolf:2021kjc}). The equations of motion are
\begin{subequations}
\begin{align}
&R_{ab}-\frac{R}{2}g_{ab}-\frac{g_{ab}}{\ell^2}= \nabla_a \Phi \nabla_b \overline{\Phi}+\nabla_a \overline{\Phi} \nabla_b \Phi-g_{ab}\nabla_c \Phi \nabla^c \overline{\Phi}\,,
\\
& \Box \Phi=0\,.
\end{align}
\end{subequations}

We now introduce double null coordinates $(U,V,\varphi)$, with $\varphi\sim \varphi+2\pi$, and take the complex scalar field $\Phi$ to have a simple harmonic dependence on $\varphi$, given by
\begin{equation}
\Phi(U,V,\varphi)=e^{{\rm i}\,m\,\varphi}\,\Xi(U,V)\,,
\end{equation}
where $m\in\mathbb{Z}$. The resulting energy-momentum tensor is independent of $\varphi$ hence the metric may be assumed to be axisymmetric, with Killing vector $\partial/\partial \varphi$. Using some residual freedom\footnote{
For example a term $-W_V(U,V)dV$ inside the square brackets can be eliminated by a shift $\phi \rightarrow \phi + \Phi(U,V)$ for suitable $\Phi$.} in defining double null coordinates the metric can be written
\begin{equation}
{\rm d}s^2=-\Omega^2(U,V){\rm d}U\,{\rm d}V+r(U,V)^2\left[{\rm d}\varphi-W_U(U,V){\rm d}U\right]^2\,.
\label{eq:lineUVvarphi}
\end{equation}
Solutions of this form break axisymmetry only in the matter sector and have been used in \cite{Dias:2025uyk,Dias:2026xuy} to construct hairy black holes with AdS$_3$ asymptotics with non-trivial boundary conditions.

Since the metric is axisymmetric, the results of section \ref{sec:lemproof} imply that, for a circle of constant $U,V$, in both R and NS sectors, our quasilocal energy coincides with the renormalized Hawking mass $\varpi$ and our quasilocal angular momentum $J$ coincides with the Komar angular momentum. For the above metric these are given by (henceforth we use units $8G=1$)
\be
\label{eq:JW}
   J(U,V) = \frac{2r^3}{\Omega^2}\partial_VW_U
\ee
and
\be
\varpi(U,V) = \frac{r^2}{\ell^2}+\frac{4\partial_V r\partial_Ur}{\Omega^2} + \frac{J^2}{4r^2}
\ee
We will also sometimes make use of the (unrenormalized) Hawking mass $m_H$ defined by
\begin{equation}
m_{\rm H}(U,V) \equiv \varpi -\frac{J^2}{4r^2} = \frac{r^2}{\ell^2}+\frac{4\partial_V r\partial_Ur}{\Omega^2}.
\end{equation}
Unlike $\varpi$, this is {\it not} constant in the BTZ spacetime but it agrees with the BTZ mass when evaluated for a circle at infinity.

The equations of motion decompose into several sectors:
\begin{subequations}
\paragraph{Raychaudhuri's equations}
\begin{align}
&\partial_V\left(\frac{\partial_V r}{\Omega^2}\right)+\frac{2}{\Omega^2}\left|\partial_V \Xi\right|^2r=0
\label{eq:1a}
\\
&\partial_U\left(\frac{\partial_U r}{\Omega^2}\right)+\frac{2}{\Omega^2}\left|\mathcal{D}_U \Xi\right|^2r=0
\label{eq:1b}
\end{align}

\paragraph{The angular momentum equations}
\begin{align}
&\partial_V J=4\,m\,{\rm Im}(\overline{\Xi}\,\partial_V \Xi)\,r
\label{eq:j}
\\
&\partial_U J=4\,m\,{\rm Im}(\Xi\,\overline{\mathcal{D}_U\Xi})\,r
\label{eq:jU}
\end{align}

\paragraph{The wave equations}
\begin{align}
&\partial_U\partial_Vr-\frac{J^2\Omega^2}{8r^3}-\frac{m^2|\Xi|^2}{2r}\Omega^2+\frac{r\Omega^2}{2\ell^2}=0
\label{eq:dur}
\\
& \partial_V \partial_U \Xi+\frac{\partial_V r}{2r}\mathcal{D}_U\Xi+\frac{\partial_U r}{2r}\partial_V \Xi+{\rm i}\,m\,\frac{\Omega^2J}{4r^3}\Xi+m^2\frac{\Omega^2}{4r^2}\Xi+{\rm i}\,m\,W_U \partial_V\Xi=0\,,
\label{eq:Xi}
\\
&\frac{2\partial_U \partial_V \Omega}{\Omega}-2\frac{\partial_V \Omega\,\partial_U \Omega}{\Omega^2}+\frac{m^2 |\Xi|^2\Omega^2}{2r^2}+\frac{3\Omega^2 J^2}{8r^4}+2\,{\rm Re}(\mathcal{D}_U \Xi\,\partial_V \overline{\Xi})+\frac{\Omega^2}{2\ell^2}=0
\end{align}
\label{eqs:total}
\end{subequations}
where $\mathcal{D}_U\Xi\equiv \partial_U \Xi+{\rm i}\,m\,W_U\,\Xi$. The equations of motion imply the following equations for the Hawking mass:
\begin{subequations}
\begin{align}
\partial_V m_{\rm H} &= \left(\frac{J^2}{2r^3}+\frac{2m^2 |\Xi|^2}{r}\right)\partial_V r+\frac{8 r}{\Omega^2}\left|\partial_V\Xi\right|^2(-\partial_U r),
\\
\partial_U m_{\rm H} &= -\frac{8 r}{\Omega^2}\left|\mathcal{D}_U\Xi\right|^2\partial_V r-\left(\frac{J^2}{2r^3}+\frac{2m^2 |\Xi|^2}{r}\right)(-\partial_U r).
\end{align}
\end{subequations}
From these we see that if $\partial_V r \ge 0$ and $\partial_U r \le 0$ then $m_H$ satisfies the monotonicity properties $\partial_V m_H \ge 0$ and $\partial_U m_H \le 0$.

\subsection{Characteristic gluing}

\begin{figure}[ht]
    \centering
    \includegraphics[width=0.5\linewidth]{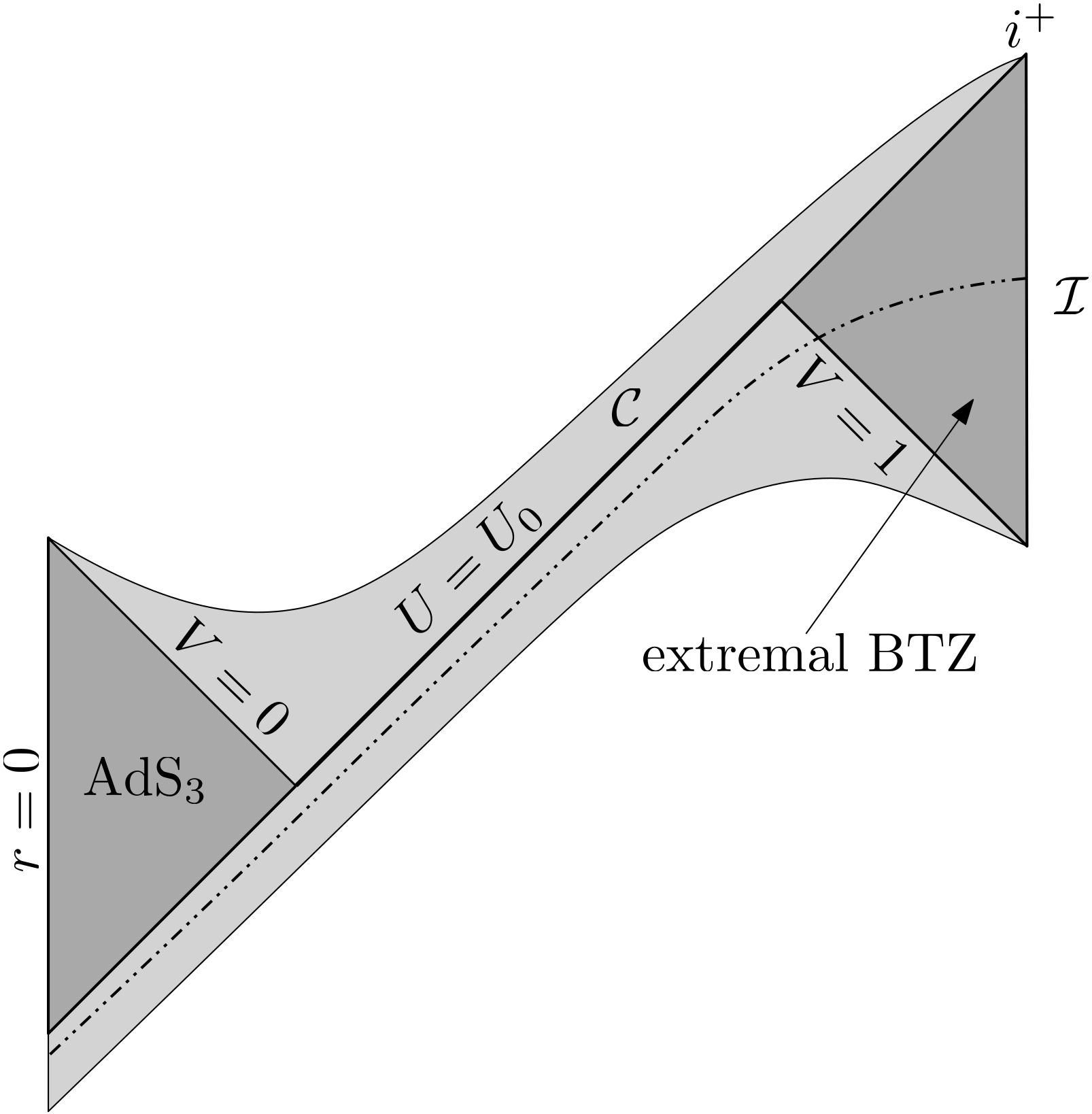}
    \caption{
    Spacetime describing formation of an extremal BTZ black hole in finite time through gravitational collapse of a scalar field. The dark regions are exactly isometric to subsets of the extremal BTZ and AdS$_3$ spacetimes. The light-shaded regions are obtained by solving characteristic initial/final value problems. These are local solutions: the construction does not determine how far they extend to the future/past. The dash-dotted line denotes a spacelike hypersurface corresponding to an instant of time before the black hole forms. The initial data on this surface describes a cloud of scalar field inside an exactly extremal BTZ region. Cauchy stability ensures the solution in the light shaded region can be extended to the centre of symmetry $r = 0$ using an argument similar to \cite{Kehle:2022uvc}.}
    \label{fig:collapse}
\end{figure}

Our aim is to construct a solution of the above equations which has a Penrose diagram of the form shown in Fig. \ref{fig:collapse}. The spacetime contains a region that is isometric to a region of AdS$_3$ containing the origin $r=0$ and another region that is isometric to the region outside the horizon of an extremal BTZ black hole. The two regions are connected by a null hypersurface $\mathcal{C}$. From the Penrose diagram it is clear that this spacetime describes a black hole and $\mathcal{C}$ is the event horizon of the black hole. The AdS$_3$ region lies inside the event horizon at early time. The entire spacetime describes gravitational collapse of a cloud of scalar field to form an extremal BTZ black hole. It is highly fine-tuned, for example the scalar field does not enter the AdS$_3$ or BTZ regions, but of course any process leading to the formation of an extremal black hole must be fine-tuned. 

Following \cite{Kehle:2022uvc} (see also~\cite{Gadioux:2025unn,Crump:2026kgu}) we will construct such a solution using {\it characteristic gluing}~\cite{gluing_apps, gluing_linear,gluing_kerr}. This is based on the characteristic initial value problem \cite{characteristic_rendall,characteristic_luk}, which involves specifying suitable data on null hypersurfaces.
We take $\mathcal{C}$ to have equation $U=U_0$. We will prescribe data on $\mathcal{C}$ that for $V \le 0$ agrees with the data on an outgoing null cone in AdS$_3$ and for 
$V \ge 1$ agrees with the data on the event horizon of an extremal BTZ black hole of radius $r_+$. 

We can now attach the AdS$_3$ and extremal BTZ regions to $\mathcal{C}$ as shown in Fig. \ref{fig:collapse}. The next step is to consider the characteristic initial value problem given by the data on the intersecting null hypersurfaces $\{U = U_0, V \ge 0\}$ and $\{V = 0, U \ge U_0\}$. The data on the latter is exactly AdS$_3$. This problem is locally well-posed~\cite{characteristic_rendall,characteristic_luk} and so the solution can be extended to the future of these hypersurfaces into the light grey region shown to the future of $\mathcal{C}$ in Fig. \ref{fig:collapse}. Similarly we can consider the characteristic {\it final} value problem given 
by the data on the intersecting null hypersurfaces $\{U = U_0, V \le 1\}$ and $\{V = 1, U \le U_0\}$. The data on the latter is exactly extremal BTZ. Well-posedness ensures that we can extend the solution to the past of these hypersurfaces into the light grey region to the past of $\mathcal{C}$ in Fig. \ref{fig:collapse}. {\it A priori} it is not clear that this solution will be regular at $r=0$, we discuss this further below. 

We have now constructed a solution to the future of $\mathcal{C}$ and another to the past of $\mathcal{C}$ that match continuously across $\mathcal{C}$. However, there is no reason that this solution will be differentiable across $\mathcal{C}$. In particular, we want the metric and scalar field to be $C^k$ across $\mathcal{C}$, with $k \ge 2$ so that the equations of motion are satisfied on $\mathcal{C}$ without having to resort to the notion of weak solutions (in fact we'll see that regularity at the origin requires $k \ge 3$). So the problem is to choose the ``free'' data in the interpolating region $0<V<1$ of $\mathcal{C}$ to ensure that the resulting solution is $C^k$ across $\mathcal{C}$. 

We aim to construct a solution for which the metric and scalar field are $C^k$ across $\mathcal{C}$. Equations \eqref{eq:j} and \eqref{eq:jU} will then imply that $J$ is $C^{k+1}$. Our metric \eqref{eq:lineUVvarphi} admits a residual gauge freedom 
\begin{equation}
\label{eq:gauge}
\widetilde{U}=f(U)\qquad \widetilde{V}=g(V)\qquad \widetilde{\Omega^2}=\frac{\Omega^2}{{f^\prime(U)} {g^\prime(V)}}\,\qquad \widetilde{W_U} = \frac{W_U}{f^\prime(U)}.
\end{equation}
which of course leaves the quasilocal quantities invariant. We use this freedom to set 
\be
\Omega(U_0,V) \equiv 1 \qquad \qquad \partial^j_U\Omega(U_0,0) = 0 \qquad (1 \leq j \leq k)
\ee
This appears to leave the freedom $g'(V) = 1/f'(U_0)$, which corresponds to a linear scaling of $V$. However, this is eliminated by our conditions that the matching to AdS$_3$ and extremal BTZ occurs at $V=0,1$. 

The free data are $r_+$, the parameter $m\in\mathbb{Z}$ and the field $\Xi$ on $\mathcal{C}$, which is chosen to be at least $C^k$ and vanish for $V<0$ and $V>1$. Once a profile for $\Xi$ is selected, we integrate \eqref{eq:1a} backwards from $V=1$ to $V=0$, subject to $r(U_0,1)=r_+$ and $\left.\partial_V r(U_0,V)\right|_{V=1}=0$ (the latter because $r$ is constant for $V \ge 1$). This determines $r(U_0,V)$ on $\mathcal{C}$, which decreases monotonically when $V$ decreases. We require $r(U_0,0)>0$: if this condition fails then we reject the chosen profile for $\Xi$. Next, using \eqref{eq:j} with $J(U_0,0)=0$, we integrate forward to determine $J(U_0,V)$. We need to choose the free data so that $J(U_0,1)=J^{\rm EBTZ}$. Now, $W_U(U_0,V)$ is determined by integrating \eqref{eq:JW}, with $W_U=0$ at $V=0$. 

Next we consider the behaviour of $\partial_U r$ along $\mathcal{C}$. We do not want any part of $\mathcal{C}$ to live inside the white hole region of the BTZ spacetime and so we require $\partial_U r<0$ on $\mathcal{C}$.\footnote{
If this condition is not satisfied then the gluing construction would produce a spacetime in which the exactly BTZ portion of the spacetime includes the full region outside the BTZ black hole and white hole horizons, as well as part of the region inside the BTZ white hole horizon. It would describe a white hole with a regular centre that collapses to form a BTZ black hole.} At $V=0$ the matching to AdS$_3$ fixes $m_H = \varpi=-1$ which determines $\partial_U r$ at $V=0$. We can then propagate $\partial_U r$ along $\mathcal{C}$ using \eqref{eq:dur}. If $\partial_U r$ does not remain negative then we reject the choice of free data.

Finally we want to impose the condition that our solution is $C^k$ across $\mathcal{C}$ so we need to examine the transverse derivatives $\partial_U^j \Omega$, $\partial_U^j r$, $\partial_U^j \Xi$ and $\partial_U^j W_U$. These quantities satisfy first order transport equations on $\mathcal{C}$ that are obtained by taking $U$-derivatives of equations of motion. Therefore they are determined either by starting at $V=0$, where they are fixed by matching to AdS$_3$, and integrating forwards or starting at $V=1$, where they are fixed by matching to extremal BTZ, and integrating backwards. We need these two choices to agree. They will agree if the forward integration gives results that match extremal BTZ at $V=1$. To do this matching, we can apply the gauge transformation \eqref{eq:gauge} (with $g'(V) = 1/f'(U_0)$) to the extremal BTZ solution in $\Omega(U_0,V)\equiv 1$ gauge: by adjusting $f'(U_0)$ and $f^{(j+1)}(U_0)$ we ensure that $\partial_U r$ and $\partial_U^j \Omega$ ($1 \le j \le k$) match the results of our forward integration at $V=1$. Similarly a gauge transformation $\phi \rightarrow \phi + \Lambda(U)$ can be applied to extremal BTZ to ensure matching of $\partial_U^j W_U$ ($0 \le j \le k$) at $V=1$. 

The matching of $\partial^j_U \Xi(U_0,V)$ at $V=1$ cannot be ensured by any residual gauge freedom: matching these quantities to their extremal BTZ values (zero) gives $k$ complex constraints on our free data, equivalently $2k$ real constraints. Equation \eqref{eq:1b} (or its $U$-derivatives) then guarantees that $\partial_U^j r$ matches its extremal BTZ value at $V=1$. 

In summary, to obtain a solution that is $C^k$ across $\mathcal{C}$ we must find free data satisfying the $2k+1$ real constraints coming from the matching of $J$ and $\partial^j_U \Xi$ as well as the conditions $r(U_0,0)>0$ and $\partial_U r<0$ on $\mathcal{C}$. 

A solution obtained this way will be $C^k$ everywhere that the coordinates are regular. However, our coordinates break down at the origin $r=0$. This implies that the solution might not be $C^k$ at $r=0$ in the region to the past of $\mathcal{C}$. Just as in \cite{Kehle:2022uvc}, the solution is expected to ``lose a derivative'' at $r=0$, where it will be only $C^{k-1}$. One way to see this is to use an argument based on Cauchy stability \cite{Kehle:2022uvc} (see also \cite{Crump:2026kgu} for a discussion in 5d, where one loses {\it two} derivatives at the origin). Thus in order to ensure that our solution is $C^2$ everywhere we will demand $k \ge 3$. 

As stated above, the free data are the integer $m$, the extremal BTZ horizon radius $r_+$, and the function $\Xi(U_0,V)$ for $V\in[0,1]$ (with $\Xi(U_0,V)=0$ for $V\notin[0,1]$). For the latter we make the ansatz
\begin{equation}
\Xi(U_0,V)
=
P_{\delta_1,\delta_2}(V)
f^{(1)}(V)
e^{{\rm i}\,V\,f^{(2)}(V)} ,
\label{eq:prof}
\end{equation}
where $f^{(1)}(V)$ and $f^{(2)}(V)$ are parametrised by single-hidden-layer neural networks with a $\tanh$ activation function,
\begin{equation}
f^{(i)}(V)
=
\chi^{(i)}_0
+
\sum_{j=1}^{p}
\chi^{(i)}_j
\tanh\left(\mu^{(i)}_j V+\kappa^{(i)}_j\right).
\label{eq:NNprofile}
\end{equation}
Here $\chi^{(i)}_0$, $\chi^{(i)}_j$, $\mu^{(i)}_j$, and $\kappa^{(i)}_j$ are free parameters. This parametrisation was recently used in \cite{Crump:2026kgu}. For further details regarding the numerical implementation and the parameter selection used to achieve the desired gluing order, we refer the reader to \cite{Crump:2026kgu}. The ansatz therefore contains a total of $6p+2$ free parameters.

The function $P_{\delta_1\,\delta_2}(V)$ ensures the required regularity at the endpoints of the gluing interval. Following the construction of \cite{Crump:2026kgu}, we choose $P_{\delta_1\,\delta_2}$ to vanish at $V=0$ and $V=1$, while remaining equal to unity in the bulk of the interval. Since in the present construction we perform a $C^3$ gluing, we take
\begin{equation}
P_{\delta_1\,\delta_2}(V)=
\begin{cases}
I_{V/\delta_1}(4,4),
& 0\leq V\leq\delta_1,\\[1mm]
1,
& \delta_1<V<1-\delta_2,\\[1mm]
I_{(1-V)/\delta_2}(4,4),
& 1-\delta_2\leq V\leq1,
\end{cases}
\label{eq:Pprofile}
\end{equation}
where $I_x(a,b)$ denotes the regularized incomplete beta function,
\begin{equation}
I_x(a,b)\equiv \frac{B(x;a,b)}{B(a,b)}.
\end{equation}
For the case relevant here,
\begin{equation}
I_x(4,4)
=
x^4\left(35-84x+70x^2-20x^3\right).
\end{equation}
Consequently,
\begin{equation}
P_{\delta_1\,\delta_2}(0)=P_{\delta_1\,\delta_2}(1)=0,
\qquad
P_{\delta_1\,\delta_2}^{(j)}(0)=P_{\delta_1\,\delta_2}^{(j)}(1)=0,
\qquad j=1,2,3,
\end{equation}
and the same vanishing of the first three derivatives holds at the transition points $V=\delta_1$ and $V=1-\delta_2$. We extend the profile by setting $P_{\delta_1\,\delta_2}(V)=0$ outside the gluing interval $V\in[0,1]$. Thus $P_{\delta_1\,\delta_2}$ is $C^3$, and piecewise $C^4$, while allowing the nontrivial gluing data to be supported entirely within $0<V<1$.

\subsection{Results}

In this section, we specialise to $k=3$ so the resulting solution is $C^2$ everywhere including at $r=0$. We are primarily interested in the existence problem, rather than in determining the minimum value of $m$ for which a given gluing can be achieved at fixed $r_+/\ell$. Following the numerical procedure outlined in \cite{Crump:2026kgu}, we find that the optimal values of $\delta_1$ and $\delta_2$ depend only weakly on $r_+/\ell$. In particular, we have found solutions for
\begin{equation}
\{r_+/\ell,m,\delta_1,\delta_2\}
=
\{0.2,4,0.1,0.1\},\qquad
\{1,5,0.05,0.1\},\qquad
\{10,400,0.05,0.1\}.
\end{equation}
For the first two cases we used $p=7$, while for the largest black hole we used $p=12$. The explicit parameter values used to construct the representative solutions shown in this paper, together with Mathematica code for reconstructing the corresponding scalar-field profiles, are available in the accompanying (\href{https://github.com/jorgealberich/BTZ_Gluing_Data}{GitHub}) online repository. We provide the full parameter sets for each of the solutions discussed in the text, using the notation of Eqs.~\eqref{eq:prof} and~\eqref{eq:NNprofile}. In our numerical investigations, we found that constructing solutions with large values of $r_+/\ell$ becomes increasingly challenging. In particular, larger values of $m$ are required. Some insight into this behaviour can be obtained from a simple scaling argument. Define
\begin{multline}
r(U,V)=r_+\,\tilde{r}\left(\frac{U}{\ell^2},V\right), \qquad
W_U(U,V)=\frac{1}{\ell r_+}\tilde{W}_{\tilde{U}}\left(\frac{U}{\ell^2},V\right),
\\
J(U,V)=\frac{r_+^2}{\ell}\tilde{J}\left(\frac{U}{\ell^2},V\right),
\qquad
\Xi(U,V)=\tilde{\Xi}\left(\frac{U}{\ell^2},V\right).
\end{multline}
The equations satisfied by the rescaled variables $\tilde{r}$, $\tilde{J}$, $\tilde{\Xi}$, and $\tilde{W}_{\tilde{U}}$ take exactly the same form as Eqs.~\eqref{eqs:total}, with $m$ replaced by $m\ell/r_+$ and $U$ replaced by $\tilde{U}=U/\ell^2$. The boundary conditions at $V=1$ become
\begin{equation}
\tilde{J}(\tilde{U}_0,1)=2, \qquad
\tilde{r}(\tilde{U}_0,1)=1, \qquad
\left.\partial_V\tilde{r}(\tilde{U}_0,V)\right|_{V=1}=0.
\end{equation}
At this stage, $m$ and $r_+$ enter only through the combination $m\ell/r_+$. Thus, given a solution for $r$ and $J$ at some value of $r_+$, one might expect to obtain a corresponding solution at a larger value of $r_+$ simply by increasing $m$ appropriately.

There is, however, one important exception to this scaling argument. Recall that $\partial_U r$ satisfies a transport equation along $\mathcal{C}$, with its initial value fixed by the matching condition for the AdS$_3$ Hawking mass. In terms of our rescaled variables this is
\begin{equation}
\left.\partial_{\tilde{U}}\tilde{r}(\tilde{U},0)\right|_{\tilde{U}=\tilde{U}_0}=-\frac{1}{4\left.\partial_V \tilde{r}(\tilde{U}_0,V)\right|_{V=0}}\left[\tilde{r}(\tilde{U}_0,0)^2+\frac{\ell^2}{r_+^2}\right]\,.
\end{equation}
Here $r_+$ appears independently of $m$. Consequently, the solution for $\partial_U r$ along $\mathcal{C}$ at different values of $r_+$ cannot be obtained simply by rescaling $m$. Moreover, increasing $r_+/\ell$ makes the right-hand side less negative. This, in turn, makes it easier for $\partial_U r$ to become positive as it is transported along $\mathcal{C}$, in which case the gluing would take place inside a BTZ white hole.

Indeed, for large values of $r_+/\ell$, we found many solutions for which $\partial_U r>0$ near $V=1$. The main numerical challenge was therefore not to find solutions satisfying $J=J^{\rm EBTZ}$, but rather to move within this family of solutions into a regime in which $\partial_U r<0$ throughout $\mathcal{C}$. To achieve this, we exploited the freedom afforded by our overparametrised scalar profile and numerically continued along the space of solutions satisfying $J=J^{\rm EBTZ}$. Starting from a converged solution, we searched for nearby solutions while biasing the continuation towards profiles for which $\partial_U r$ became progressively more negative. Repeating this procedure allowed us to move along the solution manifold until reaching configurations for which $\partial_U r<0$ along the entire gluing surface.

The profiles we find also exhibit a nontrivial dependence on $r_+/\ell$. For $r_+=0.2\,\ell$ and $m=4$ (see Fig.~\ref{fig:profilessmall}), the profiles are relatively broadly distributed along $\mathcal{C}$. By contrast, for $r_+=\ell$ and $m=5$, they closely resemble those found in \cite{Crump:2026kgu}, with most of their structure concentrated near $V=0$ (see Fig.~\ref{fig:profilesmedium}). For $r_+=10\,\ell$ and $m=400$, the profiles become delocalised once again, and instead resemble those found in \cite{Gadioux:2025unn} and employed in the pioneering construction of \cite{Kehle:2022uvc}.

It is likely that, for $r_+=10\,\ell$, one could substantially reduce $m$ and find solutions whose profiles are closer to those shown in Fig.~\ref{fig:profilesmedium}. We have not attempted a systematic exploration in this direction, since our primary goal here is to establish existence. Nevertheless, a more complete investigation of the solution space, including the dependence of the profile morphology and the minimum required value of $m$ on $r_+/\ell$, would clearly be interesting.

\begin{figure}[ht]
    \centering
    \includegraphics[width=\linewidth]{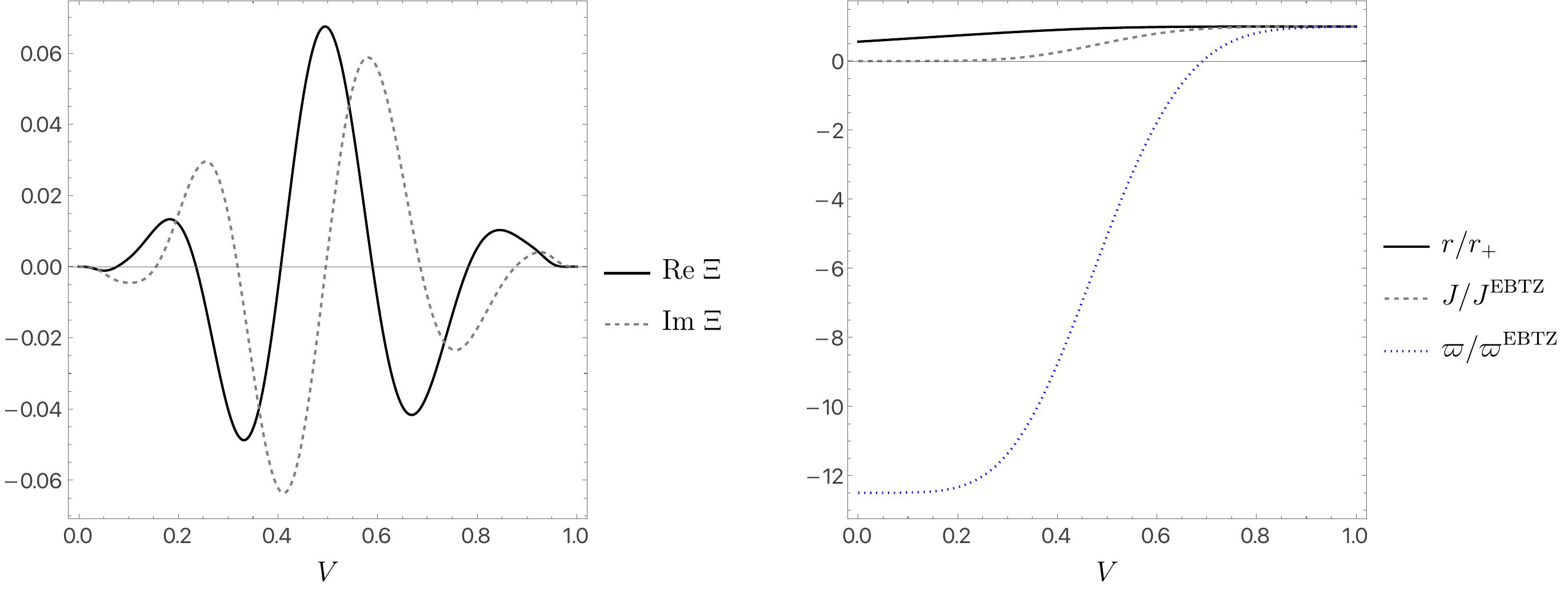}
    \caption{Left: $\operatorname{Re}\Xi$ and $\operatorname{Im}\Xi$ as functions of $V\in[0,1]$. Right: $r/r_+$, $J/J^{\rm EBTZ}$, and $\varpi/\varpi^{\rm EBTZ}$ as functions of $V\in[0,1]$. To generate these plots, we chose $m=4$ and $r_+/\ell=0.2$.}
    \label{fig:profilessmall}
\end{figure}

\begin{figure}[ht]
    \centering
    \includegraphics[width=\linewidth]{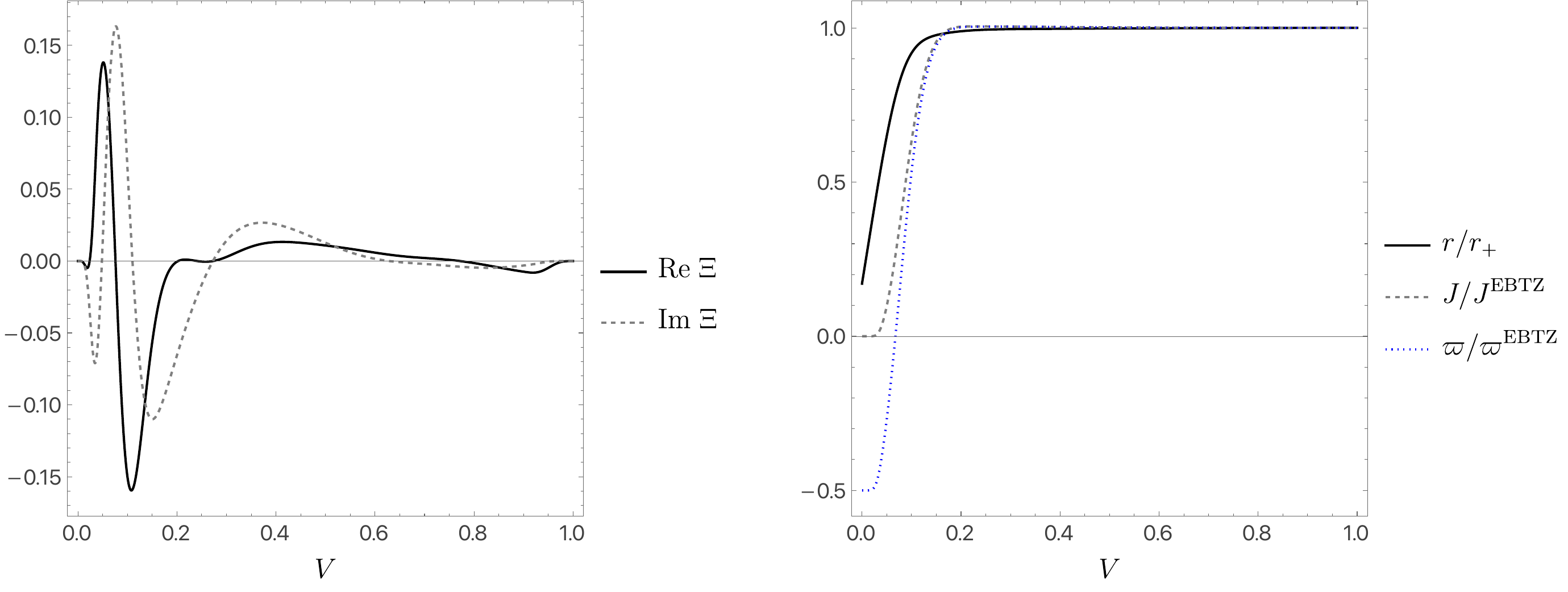}
    \caption{Left: $\operatorname{Re}\Xi$ and $\operatorname{Im}\Xi$ as functions of $V\in[0,1]$. Right: $r/r_+$, $J/J^{\rm EBTZ}$, and $\varpi/\varpi^{\rm EBTZ}$ as functions of $V\in[0,1]$. To generate these plots, we chose $m=5$ and $r_+/\ell=1$.}
    \label{fig:profilesmedium}
\end{figure}

\begin{figure}[ht]
    \centering
    \includegraphics[width=\linewidth]{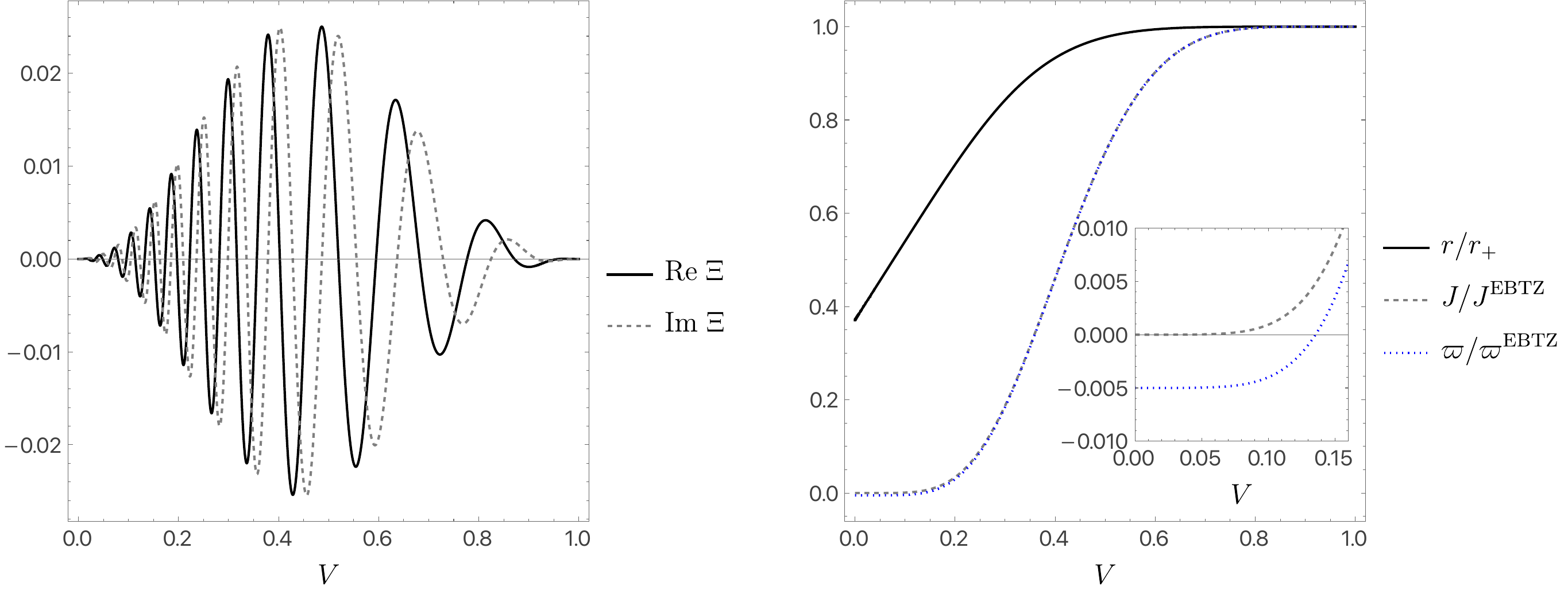}
    \caption{Left: $\operatorname{Re}\Xi$ and $\operatorname{Im}\Xi$ as functions of $V\in[0,1]$. Right: $r/r_+$, $J/J^{\rm EBTZ}$, and $\varpi/\varpi^{\rm EBTZ}$ as functions of $V\in[0,1]$. To generate these plots, we chose $m=400$ and $r_+/\ell=10$.}
    \label{fig:profileslarge}
\end{figure}

\subsection{Gluing to non-extremal BTZ}

Our Third Law guarantees that, for the theory \eqref{action}, a non-extremal BTZ black hole cannot evolve to an extremal BTZ black hole in finite time. In this section we will construct solutions in which an initial non-extremal BTZ black hole evolves in finite time to another exactly BTZ black hole that is still non-extremal but closer to extremality than the initial hole. We will investigate ``what goes wrong'' as we try to push the final black hole towards exact extremality. We will construct these solutions using gluing, as above, but replacing the initial AdS$_3$ region with a non-extremal BTZ solution. In the resulting solution, the early time part of the event horizon will be an outgoing null cone outside the apparent horizon of the initial BTZ black hole. The late time part of the event horizon coincides with the apparent horizon of the final BTZ black hole.
\begin{figure}
    \centering
    \includegraphics[width=0.7 \linewidth]{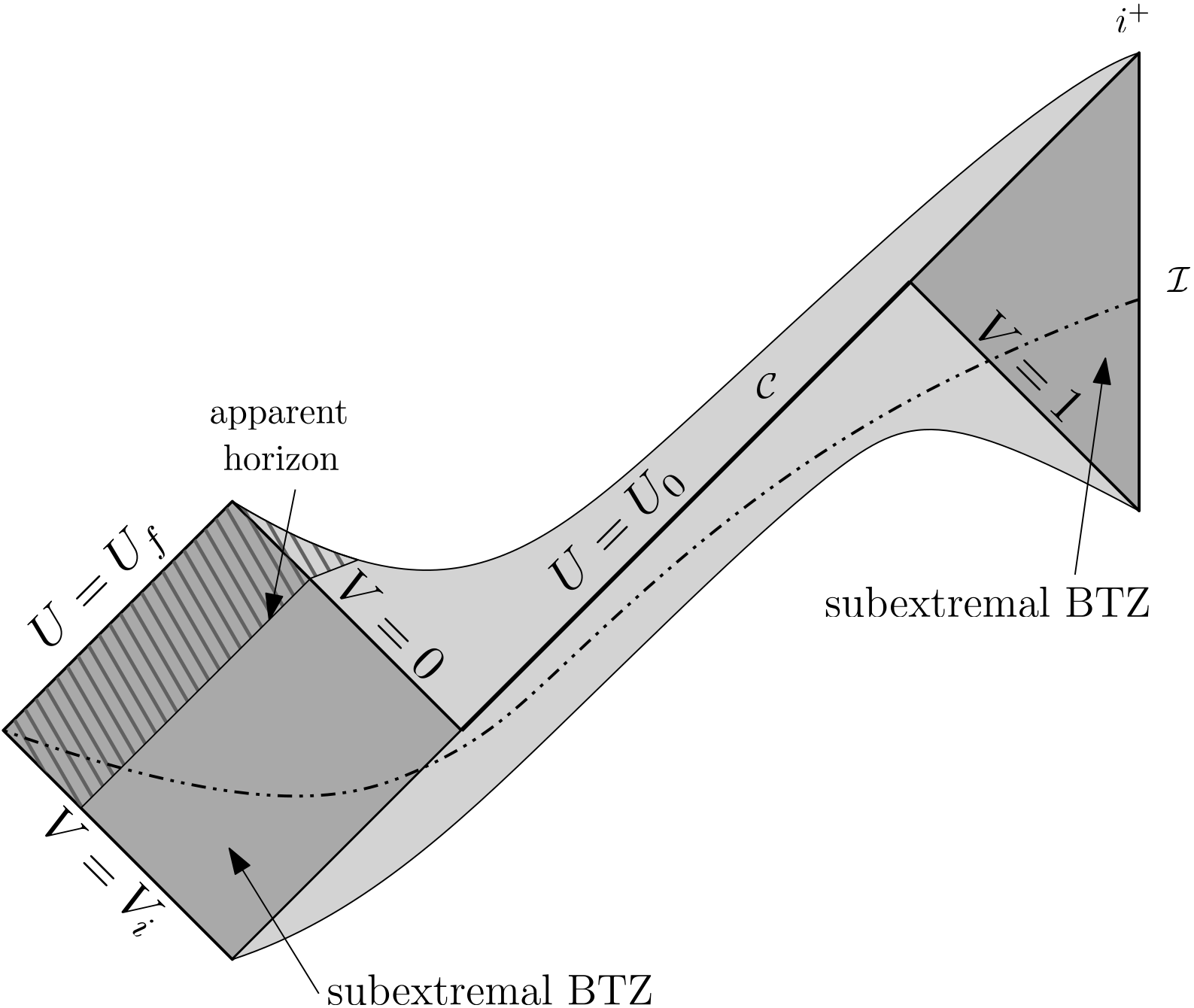}
    \caption{Spacetime in which an initial subextremal BTZ black hole evolves in finite time to a final black hole that is also exactly BTZ and closer to extremality, although still subextremal. The dash-dotted line shows an initial surface for which the induced initial data describes an exactly subextremal BTZ black hole surrounded by a region containing a non-vanishing scalar field, and then another exactly BTZ region near infinity. The hatched region shows the early time part of the region in which circles of axisymmetry are trapped; in the initial BTZ region the boundary of this region is the apparent horizon $r=r_+^-$.}
    \label{fig:subx to x}
\end{figure}
We will answer this question in the context of $C^0$ gluing. This may seem unphysical but previous work \cite{Gadioux:2025unn} has found that the results of $C^0$ gluing constructions are often qualitatively identical to results with higher regularity. Furthermore, any solution of $C^k$ gluing with $k>0$ is also a solution of $C^0$ gluing so if we can identify what goes wrong for $k=0$ then it is likely that this is also the obstruction for $k>0$. 

We denote the parameters of the initial BTZ solution as $r_+^-$ and $r_-^-$. The parameters of the final BTZ solution are denoted $r_+^+$ and $r_-^+$. The corresponding mass and angular momentum parameters, determined by \eqref{BTZMJ},  will be denoted by $\varpi^{\rm BTZ \pm}$ and $J^{\rm BTZ \pm}$. In the solutions that we construct, $r_+^-$ will be the radius of the {\it apparent} horizon of the initial black hole and $r_+^+$ the radius of both the event and apparent horizons of the final black hole. 

There are a few small changes compared to our previous gluing construction since we have now replaced the initial AdS$_3$ region with a BTZ solution and replaced the final extremal BTZ solution with a non-extremal BTZ solution. We fix $r_\pm^-$ and $r_+^+$ and solve for $r$ along $\mathcal{C}$ as before. We require that $r(U_0,0)>r_+^-$, so that $\mathcal{C}$ glues to a null cone outside the apparent horizon of the initial hole. Our initial condition for $J$ is $J(U_0,0)=J^{\rm BTZ-}$ so transporting $J$ to $V=1$ then determines $J^{\rm BTZ+}$ from which we can determine $r_-^+$. The only other difference is that we now determine the initial condition for $\partial_U r$ along $\mathcal{C}$ by demanding that $\varpi$ matches $\varpi^{{\rm BTZ}-}$ at $V=0$.

In this simpler setting of $C^0$ gluing we will use a simpler Ansatz than the one of Eq.~(\ref{eq:prof}), namely 
\begin{equation}
\Xi(U_0,V)= a_0 V(1-V) e^{{\rm i}\,m\log(\varepsilon+V)}.
\label{eq:prof-epsilon}
\end{equation}
where $\varepsilon>0$ is a new parameter. We will vary $a_0$ at fixed $(m,\varepsilon)$ and determine how closely the final black hole can approach extremality while maintaining
\begin{equation}
\left.\partial_U r(U,V)\right|_{U=U_0}<0,
\qquad
r(U_0,0)>r^-_+.
\end{equation}
For fixed $m\geq3$, as we tried to push the final black hole towards extremality by varying $a_0$, the first obstruction encountered was invariably the failure of the first of these conditions whereas for $m\leq2$, we found that the obstruction is sometimes that $r$ becomes negative as $V$ decreases, so the second condition cannot be satisfied for any choice of $r_+^-$.

Fig.~\ref{fig:com} shows results for the case $r_+^+=\ell$, $r^-_+=\ell/3$, and $r^-_-=r^-_+/2$, corresponding to an initial black hole with $J^{{\rm BTZ}-}/\ell =0.8\,\varpi^{\rm BTZ-}$. We define $J^{\rm EBTZ} = \varpi^{\rm BTZ+}\ell$: the angular momentum that the final black hole would have if it were extremal.
In the left panel of Fig.~\ref{fig:com}, we compare the largest value of $J^{\rm BTZ+}/J^{\rm EBTZ}$ attained as a function of $m$ for $\varepsilon=1$ (black disks) and $\varepsilon=10^{-1}$ (grey squares). As expected, in neither case is extremality achieved at any finite value of $m$. However, as $m \rightarrow \infty$, $J^{\rm BTZ+}/J^{\rm EBTZ}$ appears to approach a limiting value at large $m$ in both cases. For $\varepsilon=1$ this limiting value appears to be $1$. The right panel of Fig.~\ref{fig:com} shows $1-J/J^{\rm EBTZ}$ on a $\log$-$\log$ scale, making it clear that larger values of $\varepsilon$ allow us to get increasingly close to extremality at large $m$.

\begin{figure}[ht]
    \centering
    \includegraphics[width=\linewidth]{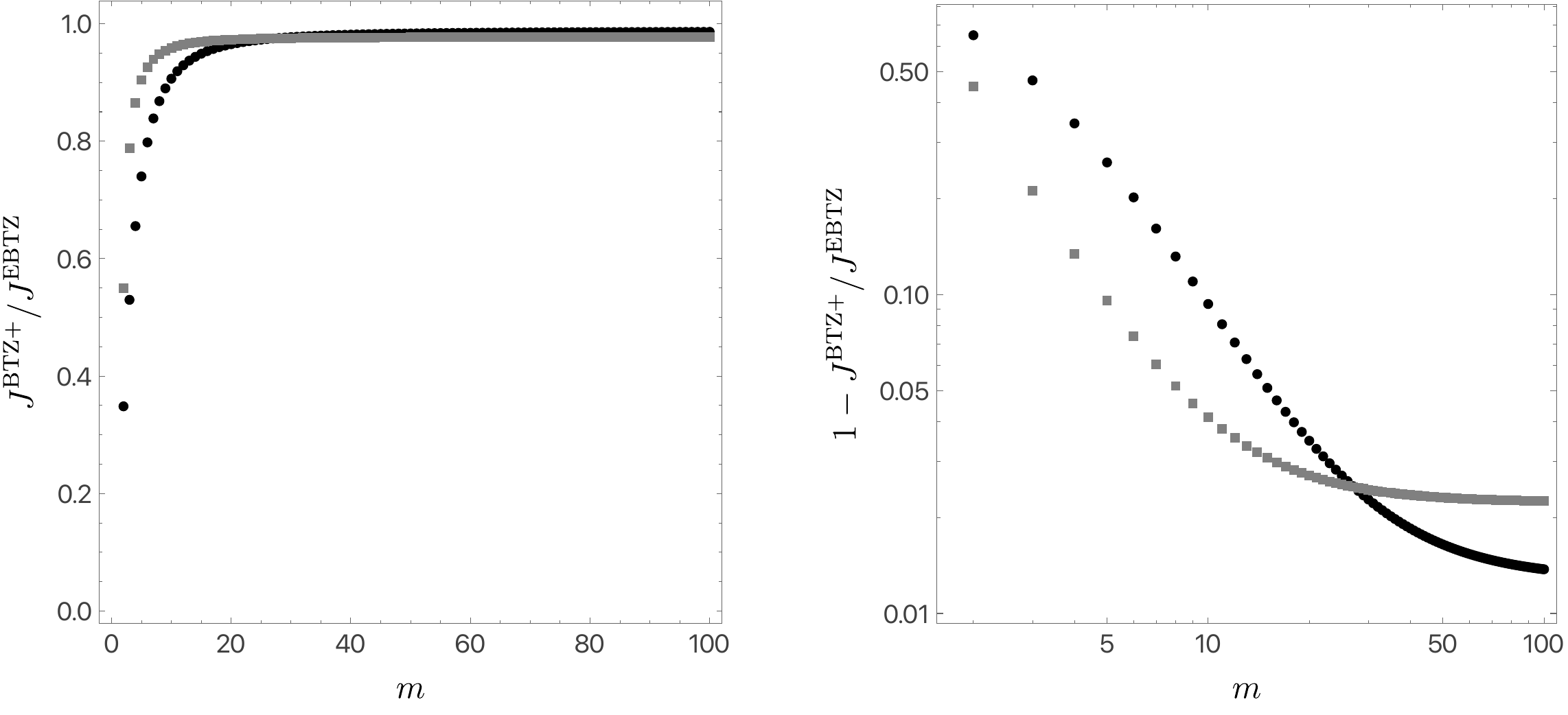}
    \caption{Comparison of the maximal angular momentum ratio $J^{\rm BTZ+}/J^{\rm EBTZ}$ as a function of $m$ for different scalar-field profiles. \textbf{Left panel:} The largest value of $J^{\rm BTZ+}/J^{\rm EBTZ}$ attained for $\varepsilon=1$ (black disks) and $\varepsilon=10^{-1}$ (grey squares). The data clearly approaches a limiting value at large $m$. \textbf{Right panel:} The quantity $1-J^{\rm BTZ+}/J^{\rm EBTZ}$ plotted on a $\log$-$\log$ scale, demonstrating that larger values of $\varepsilon$ allow the system to approach $J^{\rm BTZ+}/J^{\rm EBTZ}=1$ more closely without fully saturating. Both panels were generated using the fixed parameters $r_+^+=\ell$, $r^-_+=\ell/3$, and $r^-_-=r^-_+/2$.}
    \label{fig:com}
\end{figure}

We have constructed similar solutions for a range of initial black hole parameters. In all cases the result is the same: it is never possible to achieve exact extremality at finite $m$, in agreement with our Third Law, but one can construct sequences of solutions for which $J^{\rm BTZ+}/J^{\rm EBTZ}$ appears to tend to $1$ as $m \rightarrow \infty$. Thus we conclude that there exist solutions for which an initial non-extremal BTZ black hole can be glued to a final BTZ black hole for which $J^{\rm BTZ+}/J^{\rm EBTZ}$ can be arbitrarily close to $1$, but not exactly equal to $1$. Stated more physically, there exist solutions in which a scalar field falls into a non-extremal BTZ black hole to produce, in finite time, an exactly BTZ black hole that is also non-extremal but can be arbitrarily close to extremality. 

We have also considered the limiting case in which the initial BTZ solution has vanishing size, namely $r^-_\pm=0$. This BTZ solution has $\varpi=J=0$. It does not describe a black hole, and is singular at $r=0$, but it is of special interest because it is supersymmetric \cite{Coussaert:1993jp} and can be viewed as the Ramond sector ground state. 
The Penrose diagram for solutions produced from the gluing construction is the same as Fig.~\ref{fig:subx to x} except that no apparent horizon is present in the region of the initial BTZ solution. Our results for this case are qualitatively similar to the results of Fig.~\ref{fig:com}. In particular, we were unable to glue the $\varpi=J=0$ BTZ solution (the ``Ramond sector ground state'') to a final extremal BTZ black hole. So even though Theorem \ref{thm:3rdlaw} does not apply to this situation, the spirit of the Third Law still appears to hold here.

\subsection{Asymptotically extremal solutions}

So far, we have investigated the gluing construction of solutions that describe formation of an extremal BTZ black hole in finite time. In this section we will investigate solutions which ``settle down'' to extremal BTZ along the event horizon only asymptotically, in infinite time. There are two motivations for doing this. First, our Third Law (Theorem \ref{thm:3rdlaw}) applies only to extremal black holes formed in finite time. Could we construct a solution that initially has a trapped circle and settles down to extremal BTZ in infinite time? This might be regarded as an ``asymptotic violation'' of the Third Law. Second, Theorem \ref{thm:3rdlaw} implies that there are no trapped circles in our solutions describing gravitational collapse to form an extremal BTZ black hole in finite time. Hence these solutions might be critical solutions. However, the fact that they are exactly extremal BTZ after a finite time, rather than settling down asymptotically to extremal BTZ, suggests that they are non-generic as critical solutions. Can we find evidence for the existence of critical solutions that describe gravitational collapse to form an extremal BTZ black hole asymptotically?

We will investigate these questions using similar methodology to the gluing construction. We assume that, at late time, the behaviour of the scalar field along the event horizon will be the same as for a test field in the extremal BTZ spacetime \cite{Gralla:2019isj}. We can then use the transport equations to determine the behaviour of other quantities along the event horizon. From this we see that the quasilocal quantities on the horizon violate the R-sector BPS bound of Theorem \ref{thm:bps}, which implies that no trapped circle can be present initially. Thus we cannot violate the Third Law asymptotically using this method. Nevertheless, we can still glue an early time exactly AdS$_3$ region and hence construct a black hole interior with a regular centre and no trapped surfaces. Hence our results are consistent with existence of solutions that describe gravitational collapse to form an extremal BTZ black hole asymptotically. A limitation of this method is that we do not construct the black hole exterior, we must simply assume that there exists a solution outside the horizon that agrees with the behaviour we assume on the horizon (and with the boundary conditions at infinity). Hence we cannot claim definitively that this type of solution exists. Nevertheless, we think our results are evidence that there do exist solutions of the type described. 

We now describe these solutions in some detail. In this section, we will restrict our attention to the $C^0$ gluing. We begin by recalling the solution for a test massless scalar field on an extremal BTZ black hole with horizon radius $r_+$, obtained in \cite{Gralla:2019isj}. In double-null coordinates $(U,V,\varphi)$, it takes the form\footnote{
In the gauge $\Omega=1$, $V$ is an affine parameter. In extremal BTZ, the Killing time along the horizon is also an affine parameter. So if the black hole is settling down to extremal BTZ we expect the decay in $V$ to be the same as the decay w.r.t. Killing time determined in \cite{Gralla:2019isj}.}
\begin{subequations}
\begin{equation}
\Xi(U,V)\propto V^{-\alpha_-}\left(1-\frac{2r_+^2}{\ell^2}UV\right)^{-\alpha_+}\,,
\end{equation}
where
\begin{equation}
\alpha_{\pm}=1\pm\frac{1}{2}\frac{{\rm i}\ell m}{r_+}\,.
\end{equation}
\end{subequations}
On the future event horizon, where $U=0$ and $V>0$, one has
\begin{equation}
\Xi(0,V)\propto V^{-\alpha_-}\,.
\label{eq:lin}
\end{equation}
We shall therefore assume the same asymptotic decay in our nonlinear asymptotic gluing construction along the gluing surface $\mathcal{C}$ as $V\to+\infty$. For later convenience, we take
\begin{equation}
\mathcal{C}=\{V\geq1:U=0\}\,.
\end{equation}
As a first consistency check, we solve the $V$-Raychaudhuri equation, Eq \eqref{eq:1a}, at large $V$, assuming that, after factoring out the leading behaviour $V^{-\alpha_-}$, $\Xi$ admits an asymptotic expansion in inverse powers of $V$,
\begin{equation}
\Xi(0,V)=V^{-\alpha_-}\sum_{i=0}^{+\infty}\Xi^{(i)}V^{-i}\,.
\label{eq:Xiinf}
\end{equation}
We then find
\begin{equation}
r(0,V)=\sum_{i=0}^{+\infty}r^{(i)}V^{-i}\,,
\end{equation}
with $r^{(0)}=r_+$, $r^{(1)}=0$, and
\begin{equation}
r^{(2)}=-\frac{\ell^2m^2+4r_+^2}{12r_+}|\Xi^{(0)}|^2\,.
\end{equation}
Similarly, the angular momentum can be obtained from Eq \eqref{eq:j}, which admits the expansion
\begin{equation}
J(0,V)=\sum_{i=0}^{+\infty}J^{(i)}\,V^{-i}\,,
\end{equation}
where $J^{(0)}$ remains undetermined, $J^{(1)}=0$, and
\begin{equation}
J^{(2)}=-\ell m^2|\Xi^{(0)}|^2\,.
\end{equation}
Next, we turn our attention to the transport equation for $\partial_U r$, Eq \eqref{eq:dur}, along $\mathcal{C}$. Requiring $\partial_U r$ to remain finite as $V\to\infty$ imposes
\begin{equation}
|J^{(0)}|=\frac{2r_+^2}{\ell^2}\,,
\end{equation}
and hence requires the asymptotic black hole to be extremal. This is perhaps not surprising, since the linear expansion (\ref{eq:lin}) was obtained about a fixed EBTZ background. Indeed, we find
\begin{equation}
\left.\partial_U r(U,V)\right|_{U=0}=\sum_{i=0}^{+\infty}Z^{(i)}V^{-i}\,,
\end{equation}
where $Z^{(0)}$ remains undetermined. This is because there is residual gauge freedom arising from the freedom to rescale $U \rightarrow \alpha U$, $V-1 \rightarrow \alpha^{-1}(V-1)$ with $\alpha>0$.\footnote{This freedom was not present in our previous gluing construction because we imposed gluing conditions at both $V=0$ and at $V=1$.} Only the sign of $Z^{(0)}$ is physically meaningful. We assume that $Z^{(0)}<0$ to ensure that no antitrapped surfaces are present at large $V$. We find the solution for $Z^{(1)}$ is
\begin{equation}
Z^{(1)}=-\frac{\ell^2m^2+4r_+^2}{6 \ell^2r_+}|\Xi^{(0)}|^2\,.
\end{equation}
Asymptotically as $V\to+\infty$, these produce
\begin{equation}
\varpi-\frac{|J|}{\ell}=\frac{2}{3\,r_+}(4r_+^2+\ell^2m^2) \frac{Z^{(0)}|\Xi^{(0)}|^2}{V^3}+\mathcal{O}(V^{-4})\,.
\end{equation}
From this we see that $\varpi<|J|/\ell$ at large $V$, so a late-time cross-section $S$ of $\mathcal{C}$ violates the R-sector BPS bound. However, if a trapped surface $T$ were present inside the black hole then this BPS bound {\it would} hold on $S$, by Theorem \ref{thm:bps}.\footnote{
We can take $S$ to be at arbitrarily late time to ensure that $S$ is spacelike separated from $T$ as required in Theorem \ref{thm:bps}.} Hence there cannot be a trapped surface present. So we cannot violate the Third Law asymptotically with the above setup. 

We now turn to constructing evidence for solutions with the above asymptotics that arise from gravitational collapse. 
In order to $C^0$ glue to an exactly AdS$_3$ region at $V=1$ on $\mathcal{C}$, we need the scalar field to vanish there. One of the simplest profiles that does the trick is the one-parameter profile
\begin{equation}
\label{asymp_ansatz}
\Xi(0,V)=\Xi^{(0)}\left(1-\frac{1}{V}\right)V^{-\alpha_-}\,,
\end{equation}
which sets $\Xi^{(i)}=0$ for $i\geq2$ and $\Xi^{(1)}=-\Xi^{(0)}$ in Eq.~(\ref{eq:Xiinf}).

For fixed $m$ and $r_+/\ell$, we determine the solution by a shooting procedure in $\Xi^{(0)}$. For each trial value of $\Xi^{(0)}$, we first solve the Raychaudhuri equation, Eq. \eqref{eq:1a}, backwards from large $V$ towards $V=1$, imposing the asymptotic conditions $\{r(0,V),\partial_V r(0,V)\}\to\{r_+,0\}$ as $V\to+\infty$. This uniquely determines $r$ along $\mathcal{C}$. We then impose $J(0,1)=0$, as required for the gluing to an exactly AdS$_3$ region, and integrate Eq \eqref{eq:j} forwards from $V=1$ to large $V$. The resulting asymptotic value of $J$ is used as the shooting condition: we vary $\Xi^{(0)}$ until $\lim_{V\to+\infty}J(0,V)=J^{\rm EBTZ}$ (defined in \eqref{eq:params at extremality}). Once this condition is satisfied, we fix $\Xi^{(0)}$ and integrate Eq \eqref{eq:dur} forwards from $V=1$, imposing $\varpi(0,1)=\varpi^{{\rm AdS}_3}$. A solution is accepted as physical only if $\left.\partial_U r(U,V)\right|_{U=0}<0$ throughout $\mathcal{C}$.

Within this ansatz, and for fixed $m$, our numerical search yielded a physical $C^0$ gluing-namely, one satisfying $r(0,V)>0$ and $\left.\partial_U r(U,V)\right|_{U=0}<0$ throughout $\mathcal{C}$—only when $r_+$ was sufficiently small. In Fig.~\ref{fig:values}, we show the largest value of $r_+/\ell$ for which such a gluing was found, as a function of $m$. For every value of $m$ considered, the condition $\left.\partial_U r(U,V)\right|_{U=0}<0$ was the first to fail as $r_+/\ell$ was increased. Moreover, $\left.\partial_U r(U,V)\right|_{U=0}$ attained its maximum asymptotically as $V\to+\infty$, so that the threshold shown in Fig.~\ref{fig:values} is characterised by $\lim_{V\to+\infty}\left.\partial_U r(U,V)\right|_{U=0}=0^-$.
\begin{figure}[ht]
    \centering
    \includegraphics[width=0.75\linewidth]{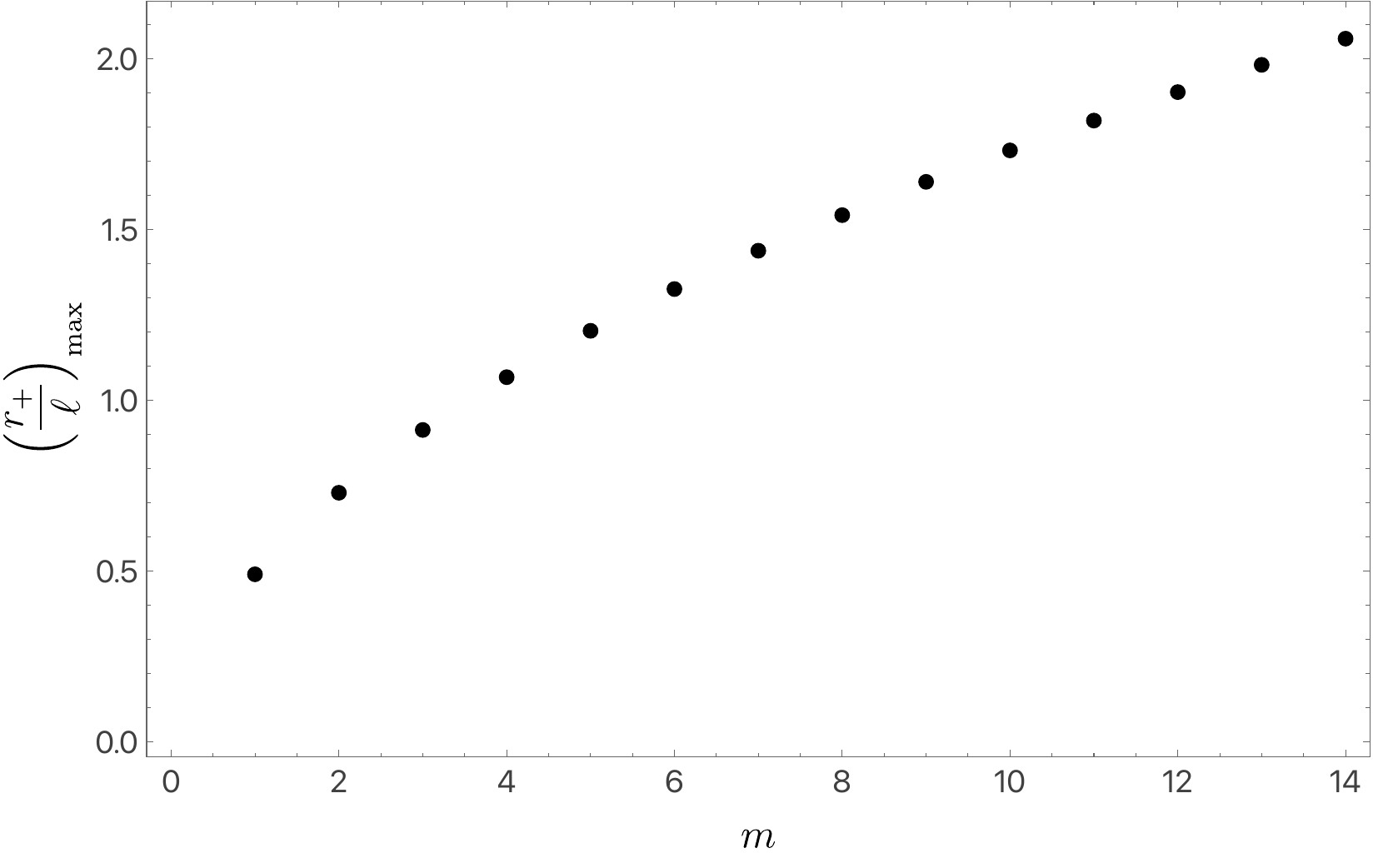}
    \caption{Maximum value of $r_+/\ell$ for which we find a physical $C^0$ gluing to AdS$_3$ with the Ansatz \eqref{asymp_ansatz}, as a function of $m$. For larger values of $r_+/\ell$, the condition $\left.\partial_U r(U,V)\right|_{U=0}<0$ is violated at large $V$.}
    \label{fig:values}
\end{figure}
In Fig.~\ref{fig:example}, we show the properties of an example solution as a function of $1/V$. 

\begin{figure}[ht]
    \centering
    \includegraphics[width=0.75\linewidth]{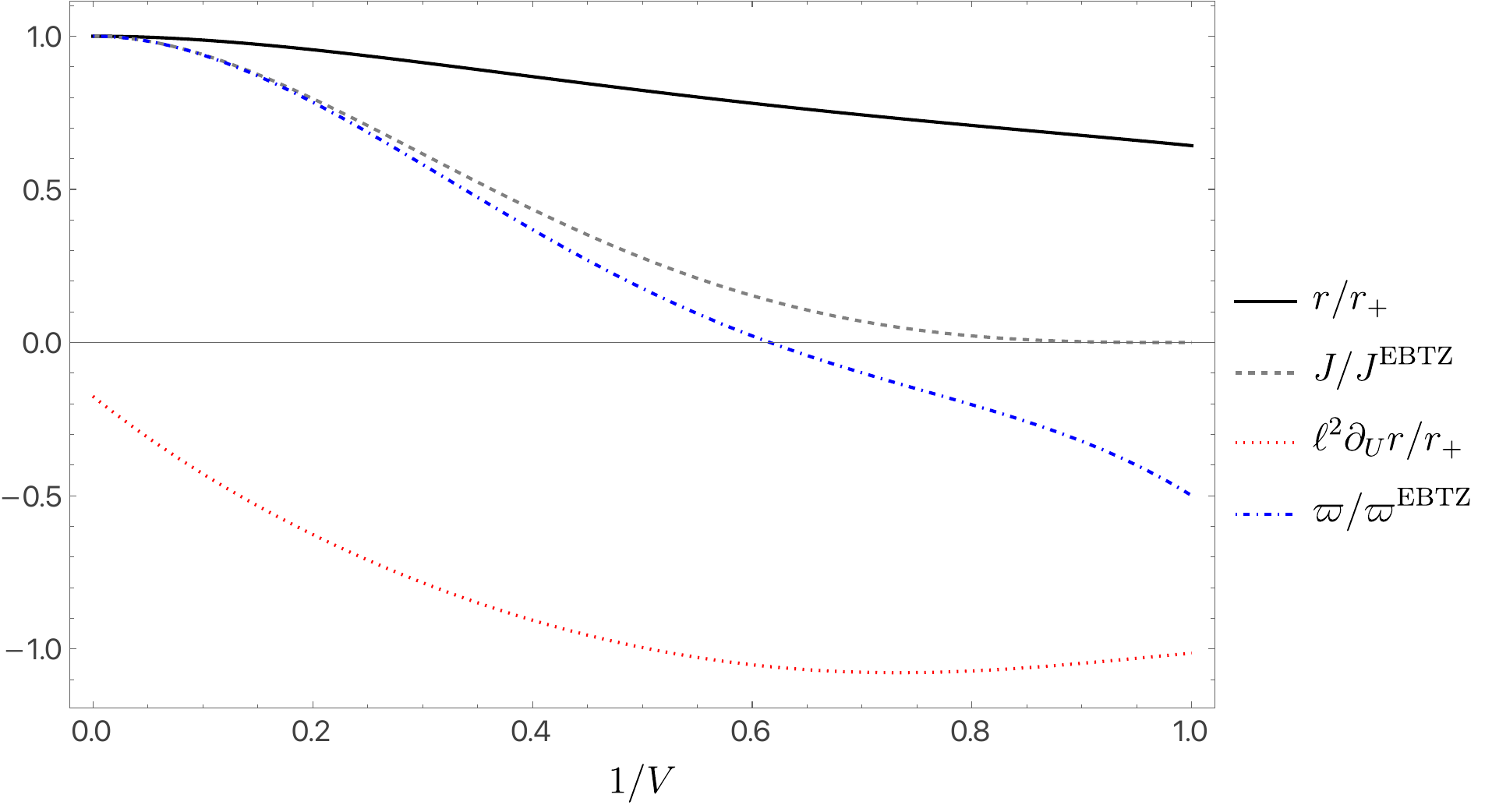}
    \caption{Properties of a physical $C^0$ gluing with $r_+=\ell$ and $m=4$, shown as a function of $1/V\in[0,1]$, so the left edge of the plot corresponds to $V=\infty$ and the right edge to where the solution is glued to AdS$_3$. We show $r/r_+$, $J/J^{\rm EBTZ}$, $\ell^2\partial_U r/r_+$, and $\varpi/\varpi^{\rm EBTZ}$, with the corresponding line styles and colors indicated in the legend on the right. The shooting parameter is $\Xi^{(0)}\approx0.92952657$, for which $\lim_{V\to+\infty}J(0,V)=J^{\rm EBTZ}$.}
    \label{fig:example}
\end{figure}

\section{Discussion}

\label{sec:discussion}

\subsection{Summary}

In this paper we have proved a Third Law of black hole mechanics, Theorem \ref{thm:3rdlaw}, for gravity in $2+1$ dimensions. If a black hole contains a trapped circle at some instant of time then it cannot evolve in finite time to a solution that coincides with an extremal BTZ black hole near the horizon. To prove this theorem we introduced spinorial definitions of quasilocal energy and angular momentum for gravity in $2+1$ dimensions and established that these definitions satisfy a BPS inequality under suitable circumstances. 

These results assumed that matter respects the dominant energy condition. However, supergravity theories may contain ``tachyonic'' scalar fields that violate this condition but respect the Breitenl\"oher-Freedman bound \cite{Breitenlohner:1982jf}. We expect that our results could be adapted to such theories by adjusting our definitions of quasilocal quantities so that the supercovariant derivative appearing in these definitions depends on the scalar fields in the way prescribed by the supersymmetry of the theory. 

We have used the method of characteristic gluing to demonstrate the existence of solutions describing gravitational collapse of a massless scalar field to form an extremal BTZ black hole in finite time. As discussed in the Introduction, this implies that there is a qualitative difference between $2+1$ dimensions and higher dimensions. In higher dimensions, when a Third Law holds for a certain type of extremal black hole then it is also impossible to form that black hole in finite time in gravitational collapse, and conversely. But $2+1$ dimensions is special in that a Third Law holds but, nevertheless, an extremal BTZ black hole can still form in finite time in gravitational collapse. The reason for this difference can be traced to the fact that the supersymmetry of extremal BTZ arises from a periodic (Ramond-sector) spinor whereas in a spacetime describing gravitational collapse, any globally defined spinor must be antiperiodic (Neveu-Schwarz sector). The fact that extremal BTZ is supersymmetric in the Ramond sector is the reason a Third Law holds and the fact that it is not supersymmetric in the Neveu-Schwarz sector is the reason it can form in gravitational collapse.

These solutions are candidates for critical solutions. However, a generic critical solution would only be expected to settle down to extremal BTZ asymptotically, not in finite time. We have presented evidence for the existence of such solutions, assuming that the decay of the scalar field along the event horizon is the same as that of a test field in the extremal BTZ spacetime. With the same assumption we have also ruled out ``asymptotic violations'' of the Third Law.

\subsection{Dual CFT interpretation}

For appropriate choices of bulk matter, our results should imply corresponding results for a dual CFT. Note that our results do not assume anything about boundary conditions at asymptotic infinity, other than that the BTZ solution satisfies these boundary conditions. What would be the CFT interpretation of our Third Law? A direct translation into CFT language appears impossible e.g. it is not known what CFT property corresponds to the existence of a trapped surface in the bulk. The property of ``being extremal at the horizon after a finite time'' is similarly hard to interpret in CFT language. However, we can perhaps look for a statement that is weaker than our Third Law that is easier to interpret in CFT. 

Consider, for example, the ``thermofield double'' state, defined on two copies of the CFT, for some non-zero temperature. This corresponds to the 2-sided BTZ solution in the bulk \cite{Maldacena:2001kr}. Assume that our doubled CFT is initially in this state. There are strictly trapped circles present in the bulk. Now try to turn on CFT sources to drive one copy of the CFT into a state that saturates the R-sector BPS bound on one of the boundaries in finite time, in such a way that the state continues to admit a geometric bulk description. Our Third Law (or just the BPS bound, applied to a surface $\Sigma$ extending from a trapped circle to infinity) indicates that this is impossible, no matter how much fine tuning one performs. 

On the other hand, we have constructed solutions describing gravitational collapse to form an exactly extremal BTZ black hole in finite time. Let's assume that the initial state for such a solution is obtained from the vacuum state by turning on sources in one copy of the CFT. The existence of these solutions shows that one can fine-tune the sources to drive the vacuum state into a final state that saturates the R-sector BPS bound after a finite time. But in this case the CFT must be defined in the NS sector, and the NS-sector BPS bound is not saturated.

\subsection{The other Third Law}

The Third Law of black hole mechanics is motivated by the ``unattainability of zero temperature'' version of the Third Law of thermodynamics. Another version of the latter asserts that the entropy of a system should vanish at zero temperature. The black hole analogue of this version has attracted considerable attention in recent years. Calculations based on the gravitational path integral have shown that quantum gravity corrections to the entropy become important at very low temperature \cite{Ghosh:2019rcj}, which may enforce this version of the Third Law. This raises the question of whether our solutions describing gravitational collapse to extremal BTZ would exhibit large quantum gravity effects.\footnote{The following discussion is motivated by conversations with Mihalis Dafermos, Roberto Emparan, Geoff Penington and Joaquin Turiaci.} This would be very exciting because it would be an example of classical evolution leading to a regime where quantum gravity becomes important, but without the classical theory showing any sign of breaking down. A natural question is whether quantum gravity effects become important as soon as the extremal BTZ black hole forms. 

We believe the answer to this question is no. The calculations just described apply to a state of thermal equilibrium, described by an eternal black hole which, at low temperature, has a long AdS$_2$ ``throat'' region. At zero temperature, the throat is infinite. The large quantum gravity effects are associated with fluctuations of this long, or infinite, throat. However, this long throat region is not part of the spacetime describing an extremal black hole formed in gravitational collapse. So we don't expect that quantum gravity effects will be large as soon as an extremal black hole has formed. However, they may gradually build up over time. For example, consider a foliation of the extremal BTZ region by surfaces of constant BTZ time coordinate $t$. Such surfaces develop a long throat region asymptotically at large $t$. This may indicate that quantum gravity effects become important at late time. It would be very interesting to study this further and to determine the time scale over which this effect takes place.

\subsection{Charged BTZ black holes}\label{sec:charged btz}

One can ask whether results analogous to ours hold for {\it charged} black hole solutions of general AdS$_3$ supergravities. A historically important example is the case of $(p,q)$-supergravities \cite{achucarro1986chern}, whose actions contain Chern-Simons terms for gauge groups $SO(p)$ and $SO(q)$ (see \cite{Bac:2026eqj} for a clear review). Charged black hole solutions of these theories have exactly the same metric as BTZ (equation \eqref{eq:btz metric}), labelled by $r_\pm$, but the mass and angular momentum now gain a contribution from the charges of these gauge fields. The supercovariant derivatives are modified so that they are gauge covariant. We could run through a similar story to the uncharged case: by imposing a certain condition on spinors of boost-weight $b=-1/2$ on some spacelike circle, two operators defined on that circle naturally arise.
To evaluate the eigenvalues of these on a given circle, one also needs to define a quasi-local charge associated with each gauge field.
A natural definition is that the charges are the eigenvalues of the holonomies of each connection about the circle.
For the case of a $U(1)$ gauge field, this amounts to the integral of the gauge field around the circle (up to normalisation).
In axisymmetry, these reduce to equation (2.28) of \cite{Bac:2026eqj}.
These operators and expressions for their eigenvalues in axisymmetry are presented in Appendix \ref{app:charged quantities} for the case where ${\rm SO}(p)={\rm SO}(q)={\rm U}(1)$.

We now consider whether such charged BTZ black holes satisfy the Third Law and whether they can form in collapse.
A new feature is that the Killing spinors of extremal, charged BTZ black holes can be either periodic or anti-periodic depending on the value of the charge.
In particular, the charge associated to each gauge field satisfies $Q^{\pm}\in \mathbb{Z}+\delta^{\rm R/NS}/2$ (for one of $\pm$) \cite{Banados:2015tft}.
When $Q^\pm$ is half-integer the black holes are supersymmetric in the Neveu-Schwarz sector, while when $Q^\pm$ is integer they are supersymmetric in the Ramond sector.
Restricting to uncharged matter satisfying the dominant energy condition, we can once again prove a Third Law analogous to \ref{thm:3rdlaw}. 

With uncharged matter one cannot form a charged black hole in gravitational collapse. Proving a Third Law with charged matter would require imposing some local bound on the electric current of matter in terms of its energy-momentum tensor, dictated by the supersymmetry of the theory, as was done for 4d theories in Refs \cite{Reall:2024njy, McSharry:2025iuz}. By carrying out a similar analysis, we expect that those charged extremal BTZ black holes which are supersymmetric in the Neveu-Schwarz sector cannot form in collapse of matter satisfying this local bound, whereas those supersymmetric in the Ramond sector could well do so.
Moreover, we expect that the Third Law would hold in both settings.

\acknowledgments

HSR and JES are supported by STFC grant No. ST/X000664/1. JES is also partially
supported by Hughes Hall College. AMM is supported by an STFC studentship and thanks C. B\"ar, A. Bac and V. Rallabhandi for helpful discussions. 

\appendix

\section{Alternative quasilocal quantities}

\label{sec:alternate_defs}

The quasi-local mass and angular momentum of section \ref{sec:quasilocal quantities} were defined on closed, co-dimension two submanifolds with negative ingoing expansion, and no restriction on the outgoing expansion.
On a surface $S$ with $\theta_{\rm out}>0$ and unrestricted ingoing expansion, we can formulate alternative definitions of these quantities in terms of an analogous operator.
Recall that the operators ${\cal O}^\pm$ were obtained by imposing equation \eqref{hol1} on $\psi_+$ in equation \eqref{ICexpr} and then integrating by parts.
Alternatively, we could impose
\begin{equation}
\label{antihol}
    \frac{1}{2}\left(1-\epsilon\right)\nabla_1^\pm \psi \equiv \Ds\psi_+ - \frac{1}{2} \theta_{\rm out} \psi_- \mp \frac{i}{2\ell}\Gamma^2\psi_+ = 0.
\end{equation}
Under the assumption that $\theta_{\rm out}\neq0$ on $S$, this uniquely determines $\psi_-$ in terms of $\psi_+$ on $S$.
As before, we can define a boost-invariant inner product on spinor fields with $b=+1/2$ as
\begin{equation}\label{eq:alternative boost-inv ip}
    (\eta,\phi)\equiv \int_S\frac{2}{\theta_{\rm out}}\eta^\dagger \phi.
\end{equation}
Then, we can write
\begin{equation}
    I_S^{\pm}[\psi] = \left(\psi_+,{\cal A}^\pm \psi_+\right),
\end{equation}
where
\begin{equation}
    {\cal A}^{\pm}  = \left({\cal D} \pm \frac{i}{2\ell}\Gamma^1\right)^\ddagger \left({\cal D}\pm\frac{i}{2\ell}\Gamma^1\right) + \frac{\theta_{\rm in}\theta_{\rm out}}{4}.
\end{equation}
As before, $\ddagger$ represents the formal adjoint in \eqref{eq:alternative boost-inv ip}:
\begin{equation}
    \left({\cal D}\pm \frac{i}{2\ell}\Gamma^1\right)^\ddagger \equiv -{\cal D} + \frac{1}{\theta_{\rm out}}{\cal D}\theta_{\rm out} \pm \frac{i}{2\ell} \Gamma^1.
\end{equation}
These operators share many of the same properties as their cousins ${\cal O}^\pm$: they have real eigenvalues, commute with $\epsilon$ and are elliptic of Laplace-type.
Similarly to ${\cal O}^\pm$, we place special attention on the lowest eigenvalues of these operators, which we denoted by
\begin{equation}
    \begin{split}
        {\cal A}^\pm \psi_+ &= \alpha^{\rm R}_\pm \psi_+ \qquad \qquad  \psi_+ \; {\rm periodic} \nonumber \\
        {\cal A}^\pm \psi_+ &= \alpha^{\rm NS}_\pm \psi_+  \qquad \qquad \psi_+ \; {\rm antiperiodic}
    \end{split}
\end{equation}
and we define the non-extremalities associated with each sector as
\begin{equation}
    \tilde\Delta^{\rm R}_\pm = r^2 \alpha^{\rm R}_\pm, \qquad \qquad  \tilde\Delta^{\rm NS}_\pm = r^2 \alpha^{\rm NS}_\pm.
\end{equation}
We define the quasi-local energy and angular momentum as
\begin{equation}
    e^{\rm R/NS} \equiv \frac{1}{4G}\left(\delta_+^{\rm R/NS}+\delta_-^{\rm R/NS}\right)- \frac{\delta^{\rm R/NS}}{8G},\qquad j^{\rm R/NS}\equiv \frac{\ell}{4G}\left(\delta_+^{\rm R/NS}-\delta_-^{\rm R/NS}\right)\,.
\end{equation}
Under the same topological assumptions on $S$ as in the statement of theorem \ref{thm:bps} (i.e. compatibility of spin structures), but replacing the convexity condition on $S$ with $\theta_{\rm out}>0$, one can prove analogous BPS bounds on these quasi-local quantities, with the same associated rigidity\footnote{The only change to the proof is the boundary conditions one imposes on the auxiliary Witten spinor encountered in section \ref{sec:proof of bps}, which become more akin to those imposed on anti-trapped surfaces after considering the inhomogeneous problem with homogeneous boundary conditions.}.

\section{Quasilocal operators in the presence of charge}\label{app:charged quantities}
As a demonstration of the discussion in section \ref{sec:charged btz}, we define a quasi-local energy, angular momentum and charge in the presence of two, flat, ${\rm U}(1)$ gauge fields.

We replace equation \eqref{supercovdef} with
\begin{equation}
    \nabla^\pm_\mu \equiv D_\mu^\pm \pm \frac{i}{2\ell}\Gamma_\mu,
\end{equation}
where $D^\pm_\mu$ are the ${\rm U}(1)$-gauge-covariant derivatives of the $+(-)$ sector.
As before, we define
\begin{equation}
    \begin{split}
            I_S^\pm[\psi] &\equiv \int_{S} \psi^\dagger \Gamma^{2} \Gamma^1  \nabla^\pm_1 \psi\\
            &= \int_{S} \left[  \psi_+^\dagger\left(  \slashed{\cal D}^\pm \psi_-  + \frac{1}{2} \theta_{\rm in}  \psi_+ \mp \frac{i}{2\ell}\Gamma^{2} \psi_-\right)+ \psi_-^\dagger \left(  \slashed{\cal D}^\pm \psi_+  - \frac{1}{2} \theta_{\rm out}   \psi_-  \mp \frac{ i}{2\ell}\Gamma^{2} \psi_+\right)\right].
    \end{split}
\end{equation}
where we have defined a gauge-and-boost covariant operator
\begin{equation}
    {\cal D}^\pm \equiv {\cal D} -iq{\cal A}^\pm, \qquad \Ds^\pm \equiv \Gamma^2\Gamma^1 {\cal D}^\pm
\end{equation}
where ${\cal A}^\pm\equiv (e_1)^a A^\pm_a$, where $q$ is the charge of the object under consideration (measured in units of the charge of the gravitini).

Imposing that $\psi$ satisfies the equation in the first bracket, then rewriting the above, we are left with the operator
\begin{equation}
    {\cal O}^{\pm}  = \left({\cal D}^\pm \pm \frac{i}{2\ell}\Gamma^1\right)^\ddagger \left({\cal D}^\pm\pm\frac{i}{2\ell}\Gamma^1\right) + \frac{\theta_{\rm in}\theta_{\rm out}}{4}.
\end{equation}
On a circle of axisymmetry, eigenvalues of these operators are determined by eigenvalues of charged versions of the intrinsic Laplacian.
We label the eigenvalues of these Laplacians by an index $n$, and denote the corresponding eigenvalues of ${\cal O}^\pm$ by $\Lambda^{\rm R/NS}_\pm(n)$.
In this setting,
\begin{equation}
    r^2 \Lambda^{\rm R/NS}_\pm(n) = \left(n+\frac{\delta^{\rm R/NS}}{2}-q GQ^\pm\right)^2 +2G\left(\varpi\pm \frac{J_{\rm Komar}}{\ell}\right),
\end{equation}
where $\varpi$ is, again, the renormalised Hawking mass of \eqref{eq:renormed hawking mass} and $J_{\rm Komar}$ the Komar angular momentum.
Additionally, we have defined
\begin{equation}
    GQ^\pm \equiv \frac{1}{2\pi}\int_S A^\pm,
\end{equation}
to be the quasi-local charge associated with the gauge fields.
We can prove a non-negativity theorem for the eigenvalues  $\Lambda^{\rm R/NS}_\pm(n)$ as in Theorem \ref{thm:bps}.

On a spacelike orbit of the axial Killing field in a charged BTZ spacetime, the renormalised Hawking mass and Komar angular momentum evaluate to the uncharged BTZ mass and angular momentum as in equation \eqref{BTZMJ} (as the line element \eqref{eq:btz metric} is unchanged by the addition of the gauge fields). Hence these satisfy the uncharged BPS bound $\varpi\ge|J_{\rm Komar}|/\ell$. Setting $q=1$, the condition that one of the operators admits a zero mode is then equivalent to the mass and angular momentum parameters in the line element satisfying the usual extremality condition for BTZ (i.e. $\varpi=|J_{\rm Komar}|/\ell$), but also that (one of) $GQ^\pm \in \mathbb{Z} + \delta^{\rm R/NS}/2$ \cite{Izquierdo:1994jz, Banados:2015tft}.


 \bibliographystyle{JHEP}
 \bibliography{biblio.bib}



\end{document}